\documentclass[a4paper,UKenglish,cleveref, autoref, thm-restate,nolineno]{socg-lipics-v2021}

\hideLIPIcs  

\usepackage[utf8]{inputenc}
\usepackage{tabularx}
\usepackage{microtype}
\usepackage{hyperref}
\usepackage{amsmath,amssymb,amsfonts,amsthm}
\usepackage{mathtools}
\usepackage{enumerate,enumitem}
\usepackage{color}
\usepackage{xspace}
\usepackage{soul}
\usepackage{tikz}
\usepackage{thm-restate}
\usepackage{booktabs}
\usetikzlibrary{shapes,shapes.callouts,shadows,arrows,backgrounds,%
matrix,patterns,arrows,decorations.pathmorphing,decorations.pathreplacing,%
positioning,fit,calc,decorations.text,arrows.meta,quotes%
}

\usepackage{cases}
\usepackage{multirow}
\usepackage{caption,subcaption}
\usepackage{todonotes}

\mathchardef\mhyphen="2D

\newcommand*{\eps}{\varepsilon}

\DeclarePairedDelimiter\norm{\lVert}{\rVert}

\newcommand*{\cF}{\mathcal{F}}

\newcommand*{\cY}{\mathcal{Y}}

\newcommand*{\FF}{\mathbb{F}}
\newcommand*{\GG}{\mathbb{G}}

\newcommand*{\NN}{\mathbb{N}}

\newcommand*{\ZZ}{\mathbb{Z}}

\newcommand*{\naturals}{\mathbb{N}}

\newcommand*{\rx}{\mathsf{x}}
\newcommand*{\ry}{\mathsf{y}}

\newcommand{\bigoh}{\mathcal{O}}
\newcommand{\bigohs}{\mathcal{O}^*}
\newcommand{\cc}[1]{{\mbox{\textnormal{\textsf{#1}}}}\xspace}  %
\newcommand{\SB}{\{\,}
\newcommand{\SM}{\;{|}\;}
\newcommand{\SE}{\,\}}

\newcommand{\NP}{\cc{NP}}
\newcommand{\FPT}{\cc{FPT}}

\newcommand{\Weft}{{\cc{W}}}
\newcommand{\W}[1]{{\Weft}{{[#1]}}}

\newcommand*{\UGC}{\textsc{UGC}\xspace}

\newcommand*{\lin}[2]{\textsc{\ensuremath{#1}-Lin(\ensuremath{#2})}}
\newcommand*{\minlin}[2]{\textsc{Min-\ensuremath{#1}-Lin(\ensuremath{#2})}}

\newcommand*{\mintwolin}{\textsc{Min-2-Lin}}
\newcommand*{\minthreelin}{\textsc{Min-3-Lin}}
\newcommand*{\minrlin}{\textsc{Min-$r$-Lin}}

\newcommand*{\mintwolinz}[1]{\textsc{Min-2-Lin(\ensuremath{\ZZ_{#1}})}}
\newcommand*{\twolinz}[1]{\textsc{2-Lin(\ensuremath{\ZZ_{#1}})}}
\newcommand*{\rep}[1]{\ensuremath{#1^{\sim}}}
\newcommand{\EQ}{\Gamma}
\newcommand*{\range}[2]{\ensuremath{\{#1,\dots,#2\}}}

\newcommand{\BA}[1]{#1_G}
\newcommand{\con}{\textsf{eqn}}
\newcommand{\sep}{\textsf{sep}}
\newcommand{\ed}{\textsf{ed}}
\newcommand{\comp}[1]{\overline{#1}}
\newcommand{\CCC}{\mathcal{C}}

\newcommand{\clasn}[1]{\tau_{#1}}

\DeclareMathOperator{\cost}{cost}
\DeclareMathOperator{\mincost}{mincost}

\DeclareMathOperator{\ord}{ord}
\DeclareMathOperator{\lsu}{lsu}
\DeclareMathOperator{\nxt}{next}

\usepackage{boxedminipage}

\newcommand{\pbDefP}[4]{%
  \noindent
  \begin{center}
  \begin{boxedminipage}{0.98 \columnwidth}
  {\sc #1}\\[5pt]
  \begin{tabular}{l p{0.70 \columnwidth}}
  {\sc Instance}: & #2\\
  {\sc Parameter:} & #3\\
  {\sc Question}: & #4
  \end{tabular}
  \end{boxedminipage}
  \end{center}
}

\newcommand*{\MLD}{\ensuremath{\textsc{MLD}_{p}}}

\usepackage[noend]{algpseudocode}
\usepackage{algorithm,algorithmicx}

\algnewcommand\algorithmicinput{\textbf{Input:}}
\algnewcommand\INPUT{\item[\algorithmicinput]}

\algnewcommand\algorithmicoutput{\textbf{Output:}}
\algnewcommand\OUTPUT{\item[\algorithmicoutput]}

\title{Parameterized Approximability for Modular Linear Equations} 

\author{Konrad K. Dabrowski}{Newcastle University,
  UK}{konrad.dabrowski@newcastle.ac.uk}{https://orcid.org/0000-0001-9515-6945}{}
  
\author{Peter Jonsson}{Link{\"o}ping University,
  Sweden}{peter.jonsson@liu.se}{https://orcid.org/0000-0002-5288-3330}
  {Supported by the Swedish Research Council (VR) under grant 2021-04371.}

\author{Sebastian Ordyniak}{University of Leeds, Leeds, United Kingdom}
{s.ordyniak@leeds.ac.uk}{https://orcid.org/0000-0003-1935-651X}
  {Supported by the
  Engineering and Physical Sciences Research Council (EPSRC, project EP/V00252X/1).}

\author{George Osipov}{Link{\"o}ping University,
  Sweden, and University of Oxford, UK}{george.osipov@pm.me}{https://orcid.org/0000-0002-2884-9837}
  {Supported by the Swedish Research Council (VR) under grant 2024-00274.}

\author{Magnus Wahlstr{\"o}m}{Royal Holloway, University of London, UK}{Magnus.Wahlstrom@rhul.ac.uk}{https://orcid.org/0000-0002-0933-4504}{}

\authorrunning{K.\,K. Dabrowski, P. Jonsson, S. Ordyniak, G. Osipov, M. Wahlstr{\"o}m} 

\Copyright{Konrad K. Dabrowski, Peter Jonsson, Sebastian Ordyniak, George Osipov, Magnus Wahlstr{\"o}m} 

\ccsdesc{Theory of computation~Parameterized complexity and exact algorithms}

\keywords{parameterized complexity, approximation algorithms, linear equations} 

\category{} 

\relatedversion{} 

\EventEditors{John Q. Open and Joan R. Access}
\EventNoEds{2}
\EventLongTitle{42nd Conference on Very Important Topics (CVIT 2016)}
\EventShortTitle{CVIT 2016}
\EventAcronym{CVIT}
\EventYear{2016}
\EventDate{December 24--27, 2016}
\EventLocation{Little Whinging, United Kingdom}
\EventLogo{}
\SeriesVolume{42}
\ArticleNo{23}

\begin{document}

\maketitle

\begin{abstract}
  We consider the $\minlin{r}{\ZZ_m}$ problem: given a system $S$ of 
  length-$r$ linear equations modulo $m$, find $Z \subseteq S$ of 
  minimum cardinality such that $S-Z$ is satisfiable. 
  The problem is \NP-hard and \UGC-hard to approximate in polynomial time
  within any constant factor even when $r = m = 2$.
  We focus on parameterized approximation with solution size as the parameter.
  Dabrowski, Jonsson, Ordyniak, Osipov and Wahlström [SODA-2023] showed that
  $\mintwolinz{m}$ is in \FPT if $m$ is prime
  (i.e.\ $\ZZ_m$ is a field),
  and it is \W{1}-hard if $m$ is not
  a prime power.
  We show that $\mintwolinz{p^n}$ is \FPT-approximable within a factor of $2$
  for every prime $p$ and integer $n \geq 2$. This implies that
  $\mintwolinz{m}$, $m \in \ZZ^+$, is \FPT-approximable within a factor of $2\omega(m)$ where $\omega(m)$ counts 
  the number of distinct prime divisors of $m$.
  The high-level idea behind the algorithm is to solve
  tighter and tighter relaxations of the problem,
  decreasing the set of 
  possible values for the variables at each step.
  When working over $\ZZ_{p^n}$ and
  viewing the values in base-$p$,
  one can roughly think of a relaxation as fixing 
  the number of trailing zeros and the
  least significant nonzero digits of 
  the values assigned to the variables.
  To solve the relaxed problem, 
  we construct a certain graph
  where solutions can be identified with 
  a particular collection of cuts.
  The relaxation may hide obstructions that will
  only become visible in the next iteration of the algorithm,
  which makes it difficult to find optimal solutions.
  To deal with this, we use a strategy based on shadow removal 
  [Marx \& Razgon, STOC-2011] to compute solutions that
  (1) cost at most twice as much as the optimum and 
  (2) allow us to reduce the set of values for all variables 
  simultaneously.
  We complement the algorithmic result with two lower bounds,
  ruling out constant-factor FPT-approximation for $\minlin{3}{R}$
  over any nontrivial ring $R$ and 
  for $\minlin{2}{R}$ over some finite commutative rings $R$.
\end{abstract}

\newpage

\section{Introduction}
\label{sec:intro}

Systems of linear equations are ubiquitous in computer science and 
mathematics~\cite{grcar2011ordinary} and
methods like Gaussian elimination 
can efficiently solve linear systems over various rings.
Equations and congruences over 
the ring of integers modulo $m$ ($\ZZ_m$)
are of central importance in number theory, but also have many
applications in computer science, including complexity theory, coding theory, cryptography, hash functions and pseudorandom generators, see e.g.~\cite{Blake:ic72,Dawar:etal:lmcs2013,Ding:etal:CRT,Yan:numbertheory}.
Linear equations over modular rings can be solved in polynomial time, but the methods are not suited to dealing with 
inconsistent systems of equations.
We consider the  $\minlin{r}{R}$ problem which asks to find an assignment to a system of linear equation over the ring $R$
that violates the minimum number of equations, and where each equation contains at most $r$
distinct variables. 
 This problem is \NP-hard even when $r = 2$ and $R$ is the simplest
 nontrivial ring $\ZZ_2$~\cite{Kolmogorov:etal:sicomp2017}.
We note that \minrlin($R$) for $r \in \NN$ and finite ring $R$ is a special case of
\textsc{MinCSP$(\Gamma)$} for a finite constraint language~$\Gamma$.
This, and the more general \textsc{Valued CSP}, have been widely studied
from many different perspectives,
e.g.~\cite{Bonnet:etal:esa2016,dabrowski2023ipec,Kim:etal:soda2023,Kolmogorov:etal:sicomp2017,OsipovW23equality,raghavendra2008optimal}.

Two common ways of coping with \NP-hardness 
are approximation and parameterized algorithms,
but neither of them seems sufficient in isolation
to deal with $\minlin{r}{\ZZ_m}$.
Even \mintwolin{} over finite fields such as $\ZZ_2$
is conjectured to be \NP-hard to approximate within any constant factor under
the Unique Games Conjecture (UGC)~\cite{khot2002power}:
see Definition~3 in~\cite{khot2016candidate} and the discussion that follows.
The natural parameter for  $\minlin{r}{\ZZ_m}$ is the cost of the optimal solution
(i.e.\ the number of equations not satisfied by it),
which we denote by $k$.
Under this parameterization, 
$\mintwolinz{m}$ is fixed-parameter tractable 
when $m$ is a prime, i.e.\ $\ZZ_m$ is a field, but \W{1}-hard when $m$ is not a prime power. Moreover,
the problem \minthreelin{} is \W{1}-hard
for every nontrivial ring~\cite{Dabrowski:etal:soda2023}.
This motivates us to study \emph{parameterized approximation} 
algorithms~\cite{Feldmann:etal:algorithms2020,Marx:tcj2008}.
This approach has received rapidly increasing interest (see, for instance,~\cite{EibenRW22ipec,Guruswami:etal:stoc2024,GuruswamiRS24ccc,Lokshtanov:etal:soda2021,lokshtanov2020parameterized,OsipovPW24pointalgebra}). 
Let $c \geq 1$ be a constant. 
A factor-$c$ {\em FPT-approximation
algorithm} takes an instance $(I, k)$, runs in
$O^*(f(k))$\footnote{
 The notation $O^*(\cdot)$ hides polynomial factors in the input size.
} time for an arbitrary computable function $f$,
either returns that there is no solution of size at most $k$ or returns that there is 
a solution of size at most $c \cdot k$. 
Thus, there is more
time to compute the solution (compared to polynomial-time
approximation) and the algorithm may output 
an oversized solution (unlike an exact FPT algorithm).\footnote{
  A decision $c$-approximation procedure for 
  $\mintwolinz{m}$ can be turned into an algorithm
  that returns a $c$-approximate solution using
  self-reducibility; see Section~\ref{sec:prelims}.}
Our main result is the following.
Let $\omega(m)$ be the number of 
distinct prime factors of $m$.

\begin{restatable}{theorem}{theoremalgo}
  \label{thm:algo} 
  For every $m \in \ZZ_+$, 
  $\mintwolinz{m}$ is \FPT-approximable within 
  $2\omega(m)$.
\end{restatable}
We complement the result with two lower bounds.
First, we show that allowing three or more
variables per equation leads to \W{1}-hardness
of constant-factor approximation.

\begin{restatable}{theorem}{theoremhard}
  \label{thm:hard}
  $\minlin{3}{R}$ over every nontrivial ring $R$
  is \W{1}-hard to approximate within any constant factor.
\end{restatable}
This result strengthens two previously known hardness results:
(i) $\minlin{3}{R}$ is \W{1}-hard~\cite{Dabrowski:etal:soda2023} and (ii) $\minlin{3}{R}$ is \NP-hard to
approximate within any constant (which can easily be derived
from~\cite{haastad2001some}).
While we focus on rings of the form $\ZZ_m$,
the result of Theorem~\ref{thm:algo} begs the questions
whether $\minlin{2}{R}$ is \FPT-approximable within a constant factor
for every finite commutative ring $R$.
We answer this question in the negative.

\begin{theorem}[See Theorem~\ref{thm:min2lin-hard} for a more detailed statement.]
  \label{thm:hard2}
  There exist finite commutative rings $R$ such that
  $\minlin{2}{R}$ is \W{1}-hard to approximate within any constant factor.
\end{theorem}
Theorems 5.2 and 6.2 in~\cite{Dabrowski:etal:soda2023}
leave open the question of whether $\mintwolinz{p^n}$ is \FPT or \W{1}-hard
for a prime $p$ and $n \geq 2$.
The answer is unknown even for the smallest such ring -- $\ZZ_4$.
While our result implies that $\mintwolinz{p^n}$ 
is \FPT-approximable within a factor of 2,
its exact parameterized complexity remains an intriguing open problem. 
Another open problem concerns the \FPT-approximability for
rings that are not modular. We have studied broader classes
of rings in a paper
available on arXiv~\cite{Dabrowski:etal:arxiv2025} but the results are inconclusive.

\subsection{FPT-approximation algorithm}
\label{ssec:intro-algo}

It is well known that every ring $\ZZ_m$ is a 
direct sum of rings of the form $\ZZ_{p^n}$,
where $p$ is a prime and $n$ is a positive integer.
To prove that $\mintwolinz{m}$ is constant-factor \FPT-approximable,
it suffices to show this for $\mintwolinz{p^n}$ (see Proposition~\ref{prop:product-approx}).
By iterative compression and branching, 
we can assume without loss of generality that
our input instances are \emph{simple}, i.e.
all binary equations are of the form
$x = r \cdot y$ with $r \in R$ and
all unary equations are crisp (i.e.\ undeletable) and of the form
$x = r$ for some $r \in R$.
Observe that the binary equations are homogeneous,
so the all-zero assignment satisfies them,
with the only obstacle being the crisp unary equations.

As a warm-up, let us describe a constant-approximation
algorithm for $\minlin{2}{R}$ when $R$ is a field, because 
it nicely illustrates our approach.\footnote{We remark that $\minlin{2}{\FF}$ for finite fields $\FF$
can be solved in FPT time~\cite{ChitnisCHPP16contract,Dabrowski:etal:soda2023,iwata2016half}.}
Let $(S, k)$ be a simple instance and
construct a graph $G = G(S)$ with vertices $x_i$ 
for every $x \in V(G)$ and $i \in R \setminus \{0\}$
and special vertices $s$ and $t$.
For unary equations in $S$,
add undeletable edges $sx_r$ if the equation
is the form $x = r$ for some $r \neq 0$,
and undeletable edges $x_rt$ for every $r \in R \setminus \{0\}$
if the equation is of the form $x = 0$ 
(to make an edge undeletable, 
it is sufficient to add many parallel copies).
For every binary equation $e$ of the form $x = r \cdot y$,
construct a bundle of edges $B_e = \{x_{ri} y_{i} : i \in R, ri \neq 0\}$
and add these edges to $G$.
This completes the construction.
There is a useful correspondence between certain
$st$-cuts in $G(S)$ and assignments to $S$.
Formally, for a set of vertices $X \subseteq V(G)$,
let $\delta(X)$ be the set of edges with exactly one
endpoint in $X$, i.e.\ $\delta(X)$ is the cut
separating $X$ and $V(S) \setminus X$.
If $U \subseteq V(G)$ is such that
$s \in U$, $t \notin U$ and there is at most one 
vertex $x_i \in U$ for any $x \in V(S)$,
we say that $\delta(U)$ is a \emph{conformal $st$-cut}.
Conformal cuts in $G$ correspond to assignments to $S$:
to construct an assignment for a conformal cut $U$,
set $x$ to $i$ if $x_i \in U$ for some $i \in R \setminus \{0\}$,
and set $x$ to $0$ otherwise;
the reverse direction of the correspondence is obvious.

Let $b(U)$ be the number of edge bundles $B_e$
intersected by $\delta(U)$.
Then $R$ being a field implies that
$b(U)$ is exactly the cost of the assignment corresponding to $U$.
More specifically, consider an equation $e = (x = r \cdot y)$,
suppose $x_i, y_j \in U$ and $B_e \cap \delta(U) = \emptyset$.
Then $i = rj$ because the edge $x_{rj} y_{j}$ is uncut
and both its endpoints are reachable from $s$, 
so the assignment corresponding to $U$ satisfies $e$.
This guarantee is represented by $B_e$ being a matching:
one may think of an edge $x_i y_j$ as 
an encoding of the constraint 
``$x = i$ if and only if $y = j$''.
To complete the algorithm for fields, it suffices
to compute a conformal cut $U$ with $b(U) \leq k$.
Observe that $b(U) \leq k$ implies $|\delta(U)| \leq (|R|-1)k$
because the bundles $B_e$ are of size $|R|-1$ and disjoint.
By a more careful analysis, one can show that, in fact,
$|\delta(U)| \leq 2k$ because a conformal cut
may only contain at most two edges from any bundle.
One obtains a $2$-approximation in 
\FPT time by branching in the style of
\textsc{Digraph Pair Cut}~\cite{kratsch2020representative}.

\begin{figure}
 \centering
 \begin{tikzpicture}[node distance=0.5cm]
    \tikzstyle{no}=[draw,circle, inner sep=1pt, fill]
    \tikzstyle{ed}=[draw,color=black, line width=1pt]

    \draw
    node[no, label=left:$x_1$] (l1) {}
    node[no, below of=l1, label=left:$x_2$] (l2) {}
    node[no, below of=l2, label=left:$x_3$] (l3) {}
    node[no, below of=l3, label=left:$x_4$] (l4) {}
    node[no, below of=l4, label=left:$x_5$] (l5) {}
    node[no, below of=l5, label=left:$x_6$] (l6) {}
    node[no, below of=l6, label=left:$x_7$] (l7) {}
    ;
    \draw
    (l1) +(2cm,0cm) node[no, label=right:$y_1$] (r1) {}
    node[no, below of=r1, label=right:$y_2$] (r2) {}
    node[no, below of=r2, label=right:$y_3$] (r3) {}
    node[no, below of=r3, label=right:$y_4$] (r4) {}
    node[no, below of=r4, label=right:$y_5$] (r5) {}
    node[no, below of=r5, label=right:$y_6$] (r6) {}
    node[no, below of=r6, label=right:$y_7$] (r7) {}
    ;

    \draw
    (l1) edge[ed] (r1)
    (l2) edge[ed] (r2)
    (l3) edge[ed] (r3)
    (l4) edge[ed] (r4)
    (l5) edge[ed] (r5)
    (l6) edge[ed] (r6)
    (l7) edge[ed] (r7)
    ;
\end{tikzpicture}
 \begin{tikzpicture}[node distance=0.5cm]
    \tikzstyle{no}=[draw,circle, inner sep=1pt, fill]
    \tikzstyle{ed}=[draw,color=black, line width=1pt]

    \draw
    node[no, label=left:$x_1$] (l1) {}
    node[no, below of=l1, label=left:$x_2$] (l2) {}
    node[no, below of=l2, label=left:$x_3$] (l3) {}
    node[no, below of=l3, label=left:$x_4$] (l4) {}
    node[no, below of=l4, label=left:$x_5$] (l5) {}
    node[no, below of=l5, label=left:$x_6$] (l6) {}
    node[no, below of=l6, label=left:$x_7$] (l7) {}
    ;
    \draw
    (l1) +(2cm,0cm) node[no, label=right:$y_1$] (r1) {}
    node[no, below of=r1, label=right:$y_2$] (r2) {}
    node[no, below of=r2, label=right:$y_3$] (r3) {}
    node[no, below of=r3, label=right:$y_4$] (r4) {}
    node[no, below of=r4, label=right:$y_5$] (r5) {}
    node[no, below of=r5, label=right:$y_6$] (r6) {}
    node[no, below of=r6, label=right:$y_7$] (r7) {}
    ;

    \draw
    (l2) edge[ed] (r1)
    (l4) edge[ed] (r2)
    (l6) edge[ed] (r3)
    (l2) edge[ed] (r5)
    (l4) edge[ed] (r6)
    (l6) edge[ed] (r7)
    ;
\end{tikzpicture}
 \begin{tikzpicture}[node distance=0.5cm]
    \tikzstyle{no}=[draw,circle, inner sep=1pt, fill]
    \tikzstyle{ed}=[draw,color=black, line width=1pt]

    \draw
    node[no, label=left:$x_1$] (l1) {}
    node[no, below of=l1, label=left:$x_2$] (l2) {}
    node[no, below of=l2, label=left:$x_3$] (l3) {}
    node[no, below of=l3, label=left:$x_4$] (l4) {}
    node[no, below of=l4, label=left:$x_5$] (l5) {}
    node[no, below of=l5, label=left:$x_6$] (l6) {}
    node[no, below of=l6, label=left:$x_7$] (l7) {}
    ;
    \draw
    (l1) +(2cm,0cm) node[no, label=right:$y_1$] (r1) {}
    node[no, below of=r1, label=right:$y_2$] (r2) {}
    node[no, below of=r2, label=right:$y_3$] (r3) {}
    node[no, below of=r3, label=right:$y_4$] (r4) {}
    node[no, below of=r4, label=right:$y_5$] (r5) {}
    node[no, below of=r5, label=right:$y_6$] (r6) {}
    node[no, below of=r6, label=right:$y_7$] (r7) {}
    ;

    \draw
    (l3) edge[ed] (r1)
    (l6) edge[ed] (r2)
    (l1) edge[ed] (r3)
    (l4) edge[ed] (r4)
    (l7) edge[ed] (r5)
    (l2) edge[ed] (r6)
    (l5) edge[ed] (r7)
    ;
\end{tikzpicture}
 \begin{tikzpicture}[node distance=0.5cm]
    \tikzstyle{no}=[draw,circle, inner sep=1pt, fill]
    \tikzstyle{ed}=[draw,color=black, line width=1pt]

    \draw
    node[no, label=left:$x_1$] (l1) {}
    node[no, below of=l1, label=left:$x_2$] (l2) {}
    node[no, below of=l2, label=left:$x_3$] (l3) {}
    node[no, below of=l3, label=left:$x_4$] (l4) {}
    node[no, below of=l4, label=left:$x_5$] (l5) {}
    node[no, below of=l5, label=left:$x_6$] (l6) {}
    node[no, below of=l6, label=left:$x_7$] (l7) {}
    ;
    \draw
    (l1) +(2cm,0cm) node[no, label=right:$y_1$] (r1) {}
    node[no, below of=r1, label=right:$y_2$] (r2) {}
    node[no, below of=r2, label=right:$y_3$] (r3) {}
    node[no, below of=r3, label=right:$y_4$] (r4) {}
    node[no, below of=r4, label=right:$y_5$] (r5) {}
    node[no, below of=r5, label=right:$y_6$] (r6) {}
    node[no, below of=r6, label=right:$y_7$] (r7) {}
    ;

    \draw
    (l4) edge[ed] (r1)
    (l4) edge[ed] (r3)
    (l4) edge[ed] (r5)
    (l4) edge[ed] (r7)
    ;
\end{tikzpicture}
 \caption{Graphs $B_e$ corresponding to equations
 $x = 1 \cdot y$, $x = 2 \cdot y$, $x = 3 \cdot y$ and
 $x = 4 \cdot y$.}
 \label{fig:z8-matchings}
\end{figure}
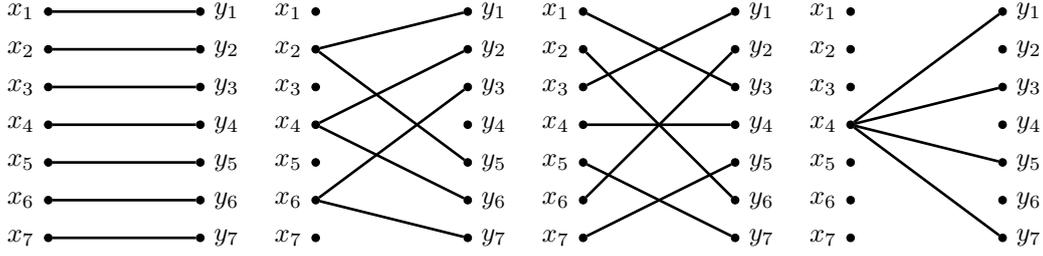

If $R$ is not a field, then
the bundles $B_e$ stop being matchings:
for example, consider 
an equation $e = (x = 2y)$ over $\ZZ_8$
(second graph from the left in Figure~\ref{fig:z8-matchings}). 
Note that both $y_2$ and $y_6$ are adjacent to $x_4$ in $B_e$.
Moreover, if $x=4$ then we require either $y=2$ or $y=6$, hence the 
dependencies cannot be captured by binary edges (even if one were to 
use directed graphs).
Thus, we lose the connection between
the number of bundles intersected by a conformal cut
and the cost of the corresponding assignment.
One idea for solving $\minlin{2}{R}$
for non-fields $R$ is to retain the 
``if and only if'' semantics of an edge in the associated graph
by matching \emph{sets} of values rather than individual values.
More specifically, let us partition $R \setminus \{0\}$ into
classes $C_1, \dots, C_\ell$ and build $G(S)$
with vertices $s$, $t$ and $x_{C_1}, \dots, x_{C_\ell}$ for all $x \in V(S)$.
One can also view the partition as an equivalence relation $\equiv$.
To keep the matching structure for an equation $e$,
we want an edge $x_{C_i} y_{C_j}$ in $B_e$ 
to mean ``$x \in C_i$ if and only if $y \in C_j$''.
For fields, we have used the most refined partition
(every nonzero element is a class of its own).
For other rings, a coarser partition is needed:
e.g. if $R = \ZZ_8$, then 
$2$ and $6$ have to be in the same class (think of
$x = 2y$).
However, simply taking a coarser partition is not sufficient:
indeed, the coarsest partition 
(putting all non-zero elements in the same class)
has the required structure,
but it only distinguishes between zero and nonzero values, 
and is not very useful algorithmically.

To define a useful equivalence relation on $\ZZ_{p^n}$, 
let us view its elements in base $p$.
Let two elements be equivalent if and only if their 
base-$p$ representations have the same number of trailing zeros
and the same least significant nonzero digit.
We denote the classes of this partition by $\EQ_{p^n}$,
and use $\EQ_{p^n}^{\neq 0}$ to refer to the nonzero classes.
For example, $\EQ_{3^2}$ (with elements written in base-$3$) 
has classes
\[
  \{01_3, 11_3, 21_3\}, 
  \{02_3, 12_3, 22_3\},
  \{10_3\}, \{20_3\},
  \{00_3\}.
\]
For another example, $\EQ_{2^3}$ has classes
\[
  \{001_2, 011_2, 101_2, 111_2\},
  \{010_2, 110_2\},
  \{100_2\}, 
  \{000_2\}.
\]
For a non-example,
a coarser partition of $\ZZ_{2^3}$ into three classes
\[
  \{001_2, 011_2, 101_2, 111_2\}, \{010_2, 110_2, 100_2\}, \{000_2\}.
\]
lacks the matching property (as evident from 
the equation $x = 2y$ in Figure~\ref{fig:z8-matchings}).

A key property of $\EQ_{p^n}$ that
allows us to make progress algorithmically 
is that it is \emph{absorbing}:
if two elements $i,j \in \ZZ_{p^n}$ are equivalent, 
then $i - j$ is divisible by $p$.
Given a class assignment $\tau$ that is guaranteed
to agree with an assignment of cost $O(k)$,
we can reduce the problem to $\mintwolinz{p^{n-1}}$.
The reduction is done by replacing every variable
$v$ with $v' = \frac{v - \rep{\tau(v)}}{p}$, where
$\rep{\tau(v)}$ is any element from the class $\tau(v)$.
If an assignment agrees with $\tau$, i.e. 
the value it assigns to $v$ belongs to $\tau(v)$,
then $v - \rep{\tau(v)}$ is divisible by $p$
by the absorption property.
Moreover, this step represents a reversible transformation of 
the solution space
(by setting $v$ equal to $pv' + \rep{\tau(v)}$),
so we do not gain or lose any solutions.

The main question is now how to compute
a class assignment that is guaranteed
to agree with a value assignment of cost $O(k)$.
Let $(S, k)$ be a simple instance and $G = G(S)$
be the corresponding class-assignment graph.
Now we have a one-to-one correspondence
between class assignments $V(S) \to \EQ^{\neq 0}$
and conformal cuts in $G$.
And again, if $(S, k)$ is a yes-instance,
then there is a conformal cut $U$ such that $b(U) \leq k$.
However, now we cannot use branching
in the style of \textsc{Digraph Pair Cut}, 
since some conformal cuts of low cost
correspond to class assignments of high cost.
For an extreme example over $\ZZ_8$,
consider the system of equations
\begin{equation}
  \label{eq:z8-example}
  \tag{$\triangle$}
  S = \{x = 4, 2a = x, 3a = b, 3b = c, 3c = a\}.
\end{equation}
The structure of the equation $x = 2y$
(see Figure~\ref{fig:z8-matchings})
implies that any matching partition of $\ZZ_{8}$
has  $2$ and $6$ in one class $A$ and matches 
it with the class $B$ containing $4$. 
By the construction of $G(S)$,
the vertex $t$ is isolated and 
the connected component $U$ of $s$ contains
vertices $s$, $x_B$, $a_A$, $b_A$ and $c_A$.
Hence, $\delta(U)$ is empty and conformal,
but the cost of $S$ is at least $1$ because
the system is inconsistent (this follows from
considering both possible values for $a$, 
which are $2$ and $6$).
In fact, the cost is exactly $1$ because
it is sufficient to delete $2a = x$.
By making vertex-disjoint copies of $S$,
we can obtain instances of arbitrarily high cost,
while the corresponding graphs
admit conformal cuts of zero size.

We also mention that the problem is not amenable
to LP-branching based on biased graphs~\cite{wahlstrom2017lp}
or the more general important balanced subgraphs
which was the main tool used in~\cite{Dabrowski:etal:soda2023}.
These methods can deal with cycle obstructions
in graphs that have the so-called {\em theta property}:
if there is a chord in a cyclic obstruction,
then at least one of the newly created cycles
must also be an obstruction.
However, our obstructions do not have this property.
Consider the instance 
\[ T = \{a = 1, a = 3b, b = 3c, c = 3d, d = 3e, e = 3a, 2a = x, x = 2c\} \]
of $\mintwolinz{8}$. 
Note that all units in $\ZZ_8$ are in the same class, and call this class $A$.
Let $B$ be the class containing $2$.
Then $G(T)$ contains a cycle on vertices $a_A$, $b_A$, $c_A$, $d_A$, $e_A$
which is an obstruction since it corresponds to 
an unsatisfiable set of equations
$\{a = 1, a = 3b, b = 3c, c = 3d, d = 3e \}$.
Moreover, there is a chord connecting $a_A$ and $c_A$ through $x_B$,
but the subinstances
$\{a = 1, a = 3b, b = 3c, 2a = x, x = 2c\}$ and
$\{a = 1, c = 3d, d = 3e, e = 3a, 2a = x, x = 2c\}$
corresponding to the newly created cycles are satisfiable.
For similar reasons, the powerful method of flow augmentation 
does not appear to directly solve the problem. 

Instead, we use a sophisticated greedy method based on shadow removal.
This technique, introduced by Marx and Razgon~\cite{marx2014fixed}
and refined in~\cite{chitnis2015directed}, provides a way
of exploring small transversals to (often implicit)
families of connected subgraphs.
Let $G$ be a graph, $\cF$ be a family of
connected subgraphs in $G$ that contain the vertex $s$,
and a positive integer $d$ be the parameter.
We are interested in transversals of $\cF$ of size at most $d$,
i.e.\ sets of at most $d$ edges that intersect every subgraph in $\cF$. 
For instance, if we are interested in $st$-cuts, 
then $\cF$ could be the set of all $st$-walks.
In our case, every conformal cut is
a transversal for all $st$-walks and all walks that
contains $s$, $x_{C_1}$ and $x_{C_2}$ for the same vertex $x$
and different classes $C_1$ and $C_2$.
More interestingly,
the conformal cut $\delta_{\sf opt} = \delta(U_{\sf opt})$ 
corresponding to an optimal class assignment is
also a transversal for 
the subgraphs $G[U_{\sf bad}]$ for $U_{\sf bad} \subseteq V(G)$ 
such that $\delta(U_{\sf bad})$ is a conformal cut,
but no satisfying assignment agrees with $U_{\sf bad}$.
The system of equations in~\eqref{eq:z8-example}
and the connected component of $s$ described below it
provide an example of a bad conformal cut.

Intuitively, shadow removal combined with 
delicate branching steps
allows us to take into account unsatisfiability 
occurring ``down the line'',
i.e.\ the obstructions that are not visible because of
the relaxation to class assignment,
but will become apparent in the next iteration.
Formally, we can compute a transversal to the same 
family of subgraphs as those intersected by $\delta_{\sf opt}$,
which are precisely all $sv$-walks with $v \notin U_{\sf opt}$.
As before, since we dismantle the bundles of edges
corresponding to an equation and treat them individually,
the size of the transversal we compute is 
$|\delta_{\sf opt}| \leq 2k$.
When translating back into equations, this means 
that we may have to delete up to $2k$ equations, 
while an optimal assignment can delete $k$.
This is the reason the cost increases.


To avoid doubling the cost at every step
and improve the overall approximation factor to $2$,
we modify the application of shadow removal to sample
a conformal cut $\delta(U)$ with the following property.
Let $Y$ be the set of equations ``intersected'' by $\delta(U)$
and $\tau$ be the class assignment corresponding to $U$.
If $S$ admits an assignment of cost $k$,
then we can guarantee that
$S - Y$ admits an assignment that agrees with $\tau$
and costs at most $k - |Y|/2$.
By committing to $\tau$, rewriting the equations over
a smaller ring and running the algorithm recursively,
we obtain a solution $Z$ to $S - Y$
of size at most $2(k - |Y|/2) = 2k - |Y|$.
Then $Y \cup Z$ is in fact a solution to $S$,
and $|Y \cup Z| \leq 2k$.


\subsection{Hardness Results}
\label{ssec:intro-hardness}

We complement the approximation algorithm
with a number of hardness results (see Section~\ref{sec:hardness}).
We prove the following (1)
$\minlin{r}{R}$ is \W{1}-hard to FPT-approximate within any constant when $r \geq 3$ and $R$ is a finite, commutative, non-trivial
ring,  and (2) there exists finite and commutative rings $R$ such
that
$\minlin{2}{R}$ is \W{1}-hard to FPT-approximate within any constant.
To prove (1), we consider the $\minlin{3}{G}$ problem
over finite Abelian groups.
Hardness for
this problem clearly carries over to finite commutative rings and
longer
equations.
Previous hardness results for 
$\minlin{3}{G}$ include the following:
Theorem~6.1 in \cite{Dabrowski:etal:soda2023} show that
$\minlin{3}{G}$ is \W{1}-hard to solve exactly and
it is known that $\minlin{3}{G}$ is not constant-factor approximable in polynomial time if
 P $\neq$ \NP (by exploiting results in~\cite{haastad2001some}).
We generalize these results by proving
that $\minlin{3}{G}$ is not FPT-approximable within
any constant if \FPT $\neq$ \W{1}.
Our result is based on a reduction from
a fundamental problem in coding
theory: the \textsc{Maximum Likelihood Decoding} problem over 
$\ZZ_p$ with $p$ a prime. 
Here we are given a matrix 
$A \in \ZZ_p^{n \times m}$, a vector $b \in \ZZ_p^m$,
and the goal is to find $x \in \ZZ_p^n$ such that
$Ax = b$ with minimum Hamming weight, i.e.
the one that minimizes $k = |\{i \in [n] : x_i \neq 0\}|$.
The parameter is $k$.
Intuitively, row $i$ in $Ax = b$ is a linear equation
$\sum_{j=1}^{n} a_{ij} x_{j} = b_{j}$, 
where $a_{ij}, b_{j} \in \ZZ_p$ are coefficients
and $x_j$ are variables.
There is a common trick to chop up long equations
into ternary equations: for example, if we have an equation
$x_1 + x_2 + x_3 + x_4 = 1$, we can introduce
auxiliary variables $y_1, y_2, y_3$ and write
\begin{align*}
  x_1 + x_2 + y_1 &= 0, \\
  y_2 + x_3 - y_1 &= 0, \\
 -y_2 + x_4 - y_3 &= 0, \\
  y_3 &= 1.
\end{align*}
Note that when summing up the new equations,
the auxiliary variables cancel out and we obtain
$x_1 + x_2 + x_3 + x_4 = 1$.
Using this trick, we encode the constraints
implied by the row equations of 
$Ax = b$ as crisp ternary and unary equations.
To encode the objective function,
i.e.\ the fact that we are minimizing 
the Hamming weight of $x$,
we add soft equations $x_j = 0$ for all $j \in [n]$.
This way, breaking a soft equation corresponds
to increasing the Hamming weight by~$1$.


We use result (1) for demonstrating that there exist finite commutative rings
such that $\minlin{2}{R}$ is \W{1}-hard to \FPT-approximate within any constant factor.
The basic idea is to express
equations of length~$3$ over $\ZZ_p$ for some suitable prime $p$ using equations of length 2 over $R$. For instance, this is
possible if we choose $R$ to be
the 16-element polynomial ring $\ZZ_2[\rx,\ry]/(\rx^2,\ry^2)$, 
i.e.\ the ring with coefficients from $\ZZ_2$
and
indeterminates $\rx,\ry$ with $\rx^2$ and $\ry^2$ factored out.
To illustrate the idea, consider an equation
$a + b + c = 0$ over $\ZZ_2$.
To express it using binary equations over $R$,
we introduce a fresh variable $v$
and consider three equations 
\begin{enumerate}
  \item $\rx v = \rx\ry b$,
  \item $\ry v = \rx\ry a$, and
  \item $(\rx+\ry) v + \rx\ry c = 0$.
\end{enumerate}
Summing up the first two equations,
we obtain $(\rx + \ry) v = \rx\ry (a + b)$.
Together with the third one,
this implies $\rx\ry (a + b + c) = 0$.
On the other hand,
any assignment that satisfies
$\rx\ry (a + b + c) = 0$
can be extended as $v = \rx a + \ry b$
to satisfy all three binary equations.
With the result for $\ZZ_2[\rx,\ry]/(\rx^2,\ry^2)$ at hand, 
it is not difficult to prove that $\minlin{2}{R}$ is 
\W{1}-hard to \FPT-approximate 
within any constant factor for many other rings $R$.

\subparagraph*{Related versions.}
A short abstract of this paper will be presented at ESA~2025.
There is also an unpublished manuscript, available on arXiv~\cite{Dabrowski:etal:arxiv2025}
about constant-factor FPT approximations for $\minlin{2}{R}$ 
for more general classes of finite commutative rings $R$. 
The existence of a constant-factor FPT-approximations for
$\minlin{2}{\ZZ_m}$ follows from results in that paper, but 
the algorithms in the current paper are more streamlined and the
constant factor is better.


\section{Preliminaries}
\label{sec:prelims}

\subparagraph*{Linear equations.}

An expression $c_1 \cdot x_1+\dots+c_r \cdot x_r=c$ is a {\em (linear) equation over $R$}  
if $c_1,\dots,c_r,c \in R$ and $x_1,\dots,x_r$ are variables with domain $R$.
This equation is {\em homogeneous} if $c=0$.
Let $S$ be a set (or equivalently a system) of equations over $R$.
We let $V(S)$ denote the variables in $S$, and we
say that $S$ is {\em consistent} if there is an assignment
$\varphi : V(S) \rightarrow R$
satisfying all equations in $S$. 
An instance of the computational problem $\lin{r}{R}$
is a system $S$ of equations in at most $r$ variables
over $R$, and the question is whether $S$ is consistent.
Linear equation systems over $\ZZ_m$ are solvable
in polynomial time
and the well-known procedure is outlined, for instance, in~\cite[p. 473]{Arvind:Vijayaraghavan:stacs2005}.
%
%
We now define the computational problem when
we allow some equations in an instance to be
soft (i.e.\ deletable at unit cost) and
crisp (i.e.\ undeletable).

\pbDefP{$\minlin{r}{R}$}
{A (multi)set $S$ of equations over $R$ with at most $r$
variables per equation,
a subset $S^\infty \subseteq S$ of crisp equations
and an integer $k$.}
{$k$.}
{Is there a set $Z \subseteq S \setminus S^\infty$ 
such that $S - Z$ is consistent and 
$|Z| \leq k$?}

We use crisp equations for convenience 
since they can be modelled by $k+1$ copies of 
the same soft equation.
For an assignment $\alpha : V(S) \to R$,
let $\cost_S(\alpha)$ be $\infty$ if $\alpha$
does not satisfy a crisp equation and
the number of unsatisfied soft equations otherwise.
We drop the subscript $S$ when it is clear from context.
We write $\mincost(S)$ to denote the minimum cost
of an assignment to $S$.

\subparagraph*{Graph Theory.}
We assume familiarity with the basics of graph theory, linear and abstract algebra, and combinatorial optimization throughout the article.
The necessary material can be found in, for instance, the textbooks by Diestel~\cite{Diestel:GT}, Artin~\cite{Artin:A} and Schrijver~\cite{Schrijver:PE}, respectively.
We use the following graph-theoretic terminology in what follows.
Let $G$ be an undirected graph.
We write $V(G)$ and $E(G)$ to denote the vertices and edges of $G$, respectively.
If $U \subseteq V(G)$, then the {\em subgraph of $G$ induced
by $U$} is the graph $G'$ with
$V(G')=U$ and $E(G')=\{ \{v, w\} \mid v,w \in U \; {\rm and} \; \{v, w\} \in E(G)\}$. 
We denote this graph by $G[U]$.
If $Z$ is a subset of edges in $G$, we write $G-Z$ to denote the graph~$G'$ with $V(G') = V(G)$ and $E(G') = E(G) \setminus Z$.
For $X,Y \subseteq V(G)$, 
an \emph{$(X,Y)$-cut} is a subset of edges $Z$
such that $G - Z$ does not contain a path with 
one endpoint in $X$ and another in $Y$.
When $X,Y$ are singleton sets $X=\{x\}$ and
$Y=\{y\}$, we simplify the notation and write $xy$-cut instead
of $(X,Y)$-cut. We say that $Z$ is an $(X,Y)$-cut \emph{closest} to
$X$ if there is no $(X,Y)$-cut $Z'$ with $|Z'|\leq |Z|$ such that the
set of vertices reachable from $X$ in $G-Z'$ is a strict subset of the
set of vertices reachable from $X$ in $G-Z$.

\subparagraph*{Parameterized Complexity.}
In parameterized 
algorithmics~\cite{book/DowneyF99,book/FlumG06,book/Niedermeier06}
the runtime of an algorithm is studied with respect to
the input size~$n$ and a parameter $p \in \NN$.
The basic idea is to find a parameter that describes the structure of
the instance such that the combinatorial explosion can be confined to
this parameter.
In this respect, the most favourable complexity class is \FPT
(\emph{fixed-parameter tractable}),
which contains all problems that can be decided by an algorithm
running in $f(p)\cdot n^{O(1)}$ time, where $f$ is a computable
function.
Problems that can be solved within such a time bound are said to be 
\emph{fixed-parameter tractable} (FPT).

We will prove that certain problems are not in $\FPT$
and this requires some extra machinery.
A {\em parameterized problem} is, formally speaking, a subset of $\Sigma^* \times {\mathbb N}$
where $\Sigma$ is the input alphabet. Reductions between parameterized problems need to take
the parameter into account. To this end, we will use {\em parameterized reductions} (or FPT-reductions).
Let $L_1$ and $L_2$ denote parameterized problems with $L_1 \subseteq \Sigma_1^* \times {\mathbb N}$
and $L_2 \subseteq \Sigma_2^* \times {\mathbb N}$. 
A parameterized reduction from $L_1$ to $L_2$ is a
mapping $P: \Sigma_1^* \times {\mathbb N} \rightarrow \Sigma_2^* \times {\mathbb N}$
such that
(1) $(x, k) \in  L_1$ if and only if $P((x, k)) \in L_2$, (2) the mapping can be computed
by an FPT-algorithm with respect to the parameter $k$, and (3) there is a computable function $g : {\mathbb N} \rightarrow {\mathbb N}$ 
such that for all $(x,k) \in L_1$ if $(x', k') = P((x, k))$, then $k' \leq g(k)$.

The class \W{1} contains all problems that are FPT-reducible to \textsc{Independent Set} when parameterized
by the size of the solution, i.e.\ the number of vertices in the independent set.
Showing \W{1}-hardness (by an FPT-reduction) for a problem rules out the existence of a fixed-parameter
algorithm under the standard assumption $\FPT \neq \W{1}$.

\subparagraph*{Parameterized Approximation.}
Let $c \geq 1$ be a constant. 
Formally, a factor-$c$ {approximation
algorithm} takes an instance $(I,k)$ of a minimization problem $\Pi$, 
accepts if $(I, k)$ is a yes-instance and 
rejects if $(I, c \cdot k)$ is a no-instance.
The running time of an \FPT-approximation 
algorithm is bounded by 
$f(k) \cdot \norm{I}^{O(1)}$ where $f: \NN \rightarrow \NN$
is some computable function. 

Note that if we have a factor-$c$ 
approximation algorithm for $\minlin{r}{R}$,
then given a yes-instance $(S, k)$,
we can produce a solution of size $ck$ 
using self-reducibility.
First, if $k = 0$, then $(S, k)$ is
a yes-instance if and only if $S$
is satisfiable.
In this case we can compute
a satisfying assignment in polynomial time (see e.g.~\cite{Arvind:Vijayaraghavan:stacs2005}).
Now, if $k > 0$,
iterate over all subsets $S' \subseteq S$
of soft equations with $|S'| = c$,
which requires polynomial time,
and for each of them check whether
$(S - S', k - 1)$ is a yes-instance,
and if so, recursively compute
an assignment $\alpha$ to $S - S'$
of cost at most $c(k-1)$.
Observe that the cost of $\alpha$ in $S$
is at most $c(k-1) + |S'| \leq c(k-1) + c = ck$.
It remains to show that
if $(S, k)$ is a yes-instance,
$(S - S', k-1)$ is a yes-instance
for some subset $S' \subseteq S$
of at most $c$ soft equations.
Indeed, if $(S, k)$ is a yes-instance,
there exists a subset 
$Z \subseteq S$ of soft equations 
such that $|Z| \leq ck$
and $S - Z$ is satisfiable.
Then for any $S' \subseteq Z$ of size $c$,
we have that $(S - S') - (Z' \setminus S')$ is satisfiable, hence 
$\mincost(S - S') \leq |Z \setminus S'| \leq ck - c = c(k-1)$.

\section{FPT-Approximation Algorithm}
\label{sec:algorithm}

We present an FPT-approximation algorithm for $\mintwolinz{m}$.
Section~\ref{ssec:algorithm-summary} contains the
``outer loop'' of the algorithm:
all the steps of the algorithm are stated
as propositions and lemmas and combined
in the proof of Theorem~\ref{thm:algo}.
The proofs of the more technically involved lemmas
are relegated to Section~\ref{ssec:comp-hom-branch},
which contains iterative compression and homogenization arguments
very similar to those in~\cite{Dabrowski:etal:soda2023},
and Section~\ref{ssec:class-assignments},
which is arguably the most challenging part of the algorithm
employing shadow removal and delicate branching.

\subsection{Algorithm Summary}
\label{ssec:algorithm-summary}

Let $p_1^{n_1} \cdots p_\ell^{n_\ell}$ be the prime factorization of $m \in \ZZ_+$.
It is well known that $\ZZ_m$ is isomorphic to the direct sum
$\bigoplus_{i=1}^{\ell} \ZZ_{p_i^{n_i}}$,  so we can reduce 
the problem to the prime power case.

\begin{proposition}
  \label{prop:product-approx}
  Suppose that the ring $R$ is 
  isomorphic to a direct sum $\bigoplus _{i=1}^{\ell} R_{i}$.
  If $\minlin{2}{R_{i}}$ is \FPT-approximable within a factor $c_i$
  for all $i \in [\ell]$, then $\minlin{2}{R}$ is \FPT-approximable
  within a factor $\sum_{i=1}^\ell c_i$.
\end{proposition}
\begin{proof}
Since $R \cong \bigoplus _{i=1}^{\ell} R_{i}$, there exists a
ring isomorphism $r \mapsto (r_1,\dots,r_\ell)$, 
where $r_i \in R_i$ for all $i \in \{1,\ldots,\ell\}$.
If $e$ is the equation $ax+by=c$ over $R$, then we
let $e_i$ denote the equation $a_i x + b_i y = c_i$ over $R_i$.
By this isomorphism, 
the equation $e$ has a solution if and only if $e_i$ has a solution for
every $i \in \{1,\ldots,\ell\}$.

Let $I=(S,k)$ denote an arbitrary instance of
$\minlin{2}{R}$ with $S=\{e^1,\dots,e^m\}$. Let 
$S_i=(\{e^1_i,\dots,e^m_i\},k)$ and $I_i=(S_i,k)$, $i \in \{1,\ldots,\ell\}$.
Compute $c_i$-approximate solutions to $I_i$, $i \in \{1,\ldots,\ell\}$.
If at least one of these instances has no solution, 
then $I$ has no solution since every such instance
is a relaxation of $I$.
Otherwise, let $\alpha_i : V \to R_i$ be
solutions to $I_i$, $i \in \{1,\dots,\ell\}$,
and define $\alpha(v) \in R$ for all $v \in V$ 
to be the preimage of
$(\alpha_1(v_1), \dots, \alpha_n(v_n))$.
Observe that $\alpha$ violates equation $e$
if and only if there is $i \in \{1,\dots,\ell\}$
such that $\alpha_i$ violates $e_i$.
Hence, $\alpha$ violates at most $k \cdot \sum_{i=1}^\ell c_i$ equations,
and we are done.
\end{proof}

Now, consider the ring $\ZZ_{p^n}$ for 
a prime $p$ and positive integer $n$.
If $n = 1$, then the ring is in fact a field,
and we can use e.g. the algorithm of~\cite{Dabrowski:etal:soda2023}
to solve the \mintwolin{} problem.
Assume $n \geq 2$.
We start with a simplification step.
An equation over a ring $R$ is \emph{simple}
if it is either a binary equation
of the form $u = rv$ for some $r \in R$
or a crisp unary equation
$u = r$ for some $r \in R$.
An instance $S$ of $\lin{2}{R}$ is simple
if every equation in $S$ is simple.
We use iterative compression,
homogenization and branching to assume that
the input instances $S$ are simple.
The details are deferred to 
Section~\ref{ssec:comp-hom-branch},
where we prove the following lemma.

\begin{restatable}{lemma}{lemsimple}
  \label{lem:simple-instances}
  Let $m$ be positive integer and $c$ be a constant. 
  If $\mintwolinz{m}$ restricted to simple instances is 
  \FPT-approximable within a factor $c$,
  then $\mintwolinz{m}$ on general instances is 
  \FPT-approximable within a factor $c$.
\end{restatable}

For the reduction step,
we will solve a relaxation of $\mintwolinz{p^n}$:
we will define an equivalence relation $\equiv$ on $\ZZ_{p^n}$
and compute an assignment of equivalence classes
(rather than values) to the variables;
then we will use this class assignment
to rewrite our input as a set of equations over $\ZZ_{p^{n-1}}$.
The partition is obtained by viewing the elements
of $\ZZ_{p^n}$ represented in base-$p$.
Formally, every element $a \in \ZZ_{p^n}$
equals $\sum_{i=0}^{p-1} a_i p^i$,
where the coefficients $a_0,\dots,a_{p-1} \in \ZZ_p$
uniquely define $a$.
Let $\vec{a} = (a_0,\dots,a_{p-1})$,
and for every $a \neq 0$, define
$\ord(a) = \min \{ i : a_i \neq 0 \}$
to be the index of the first nonzero coordinate in $\vec{a}$, and
$\lsu(a) = a_{\ord(a)}$ 
to be the least significant unit in $\vec{a}$.
For completeness, let $\ord(0) = \lsu(0) = 0$.
Let 
\[ a \equiv b \iff \ord(a) = \ord(b) \text{ and } \lsu(a) = \lsu(b). \]
This equivalence relation has two important properties.
First, it is \emph{matching}, i.e.
$\{0\}$ is an equivalence class,
and for every $i,j \in \ZZ_{p^n}$ and $r \in \ZZ_{p^n}$,
\begin{itemize}
  \item if $i \equiv j$ then $ri \equiv rj$,
  \item if $i \not\equiv j$, then either $ri \not\equiv rj$ or $ri=rj=0$.
\end{itemize}
Moreover, it is \emph{absorbing}, meaning that
\begin{itemize}
  \item $i \equiv j \implies p \text{ divides } i - j$ 
  for all $i,j \in \ZZ_{p^n}$.
\end{itemize}

Let $\EQ_{p^n}$ denote the set of equivalence
classes of $\equiv$, and 
$\EQ^{\neq 0}_{p^n} = \EQ_{p^n} \setminus \{\{0\}\}$.
We will drop the subscript when it is clear from the context.
We write $\EQ(r)$ to denote the equivalence class in $\EQ$ 
that contains the element $r$.
The name ``matching'' comes from 
considering bipartite graphs $G^{\neq 0}_r$
defined by binary equations $u = rv$ 
for every $r \in R$ as follows:
let $V(G_r) = \EQ^{\neq 0} \uplus \EQ^{\neq 0}$ 
and let there be an edge between two classes 
$C_1$ on the left and $C_2$ on the right 
if and only if $i = rj$ for some 
$i \in C_1$ and $j \in C_2$.
Then $\equiv$ being matching implies
that $G^{\neq 0}_r$ is a matching for every $r$, 
i.e.\ every vertex has degree at most~$1$.

\begin{example}
  Partition $\EQ_{3^2}$ (with elements written in base-$3$) 
  has classes
\[
  \{01_3, 11_3, 21_3\}, 
  \{02_3, 12_3, 22_3\},
  \{10_3\}, \{20_3\},
  \{00_3\}.
\]
For another example, $\EQ_{2^3}$ has classes
\[
  \{001_2, 011_2, 101_2, 111_2\},
  \{010_2, 110_2\},
  \{100_2\}, 
  \{000_2\}.
\]
For a non-example,
a coarser partition of $\ZZ_{2^3}$ into three classes
\[
  \{001_2, 011_2, 101_2, 111_2\}, \{010_2, 110_2, 100_2\}, \{000_2\}.
\]
lacks the matching property (as is evident from 
the equation $x = 2y$ in Figure~\ref{fig:z8-matchings}). 
\end{example}

The matching property of $\equiv$ is crucial
for the main algorithmic lemma, which we
state below and prove in Section~\ref{ssec:class-assignments}.
A value assignment $\alpha : V(S) \to \ZZ_{p^n}$
\emph{agrees with} a class assignment $\tau : V(S) \to \EQ$
if $\alpha(v) \in \tau(v)$ for all $v \in V(S)$.
A class assignment $\tau$ \emph{respects} an equation $e$
if it admits a satisfying assignment that agrees with $\tau$, 
otherwise we say that $\tau$ \emph{violates} $e$.
Define the cost of a class assignment $\tau$ to be
the number of equations in $S$ that it violates.
Note that every value assignment $\alpha$ 
uniquely defines a class assignment $\tau_{\alpha} : V(S) \to
\EQ_{p^n}$,
so we say that $\alpha$ \emph{strongly violates}
an equation $e$ if $\tau_{\alpha}$ violates $e$.
Clearly, an optimal assignment can violate 
at most $k$ equations in $S$, and we can guess 
the number $q \in \{0,\dots,k\}$ of strongly violated equations.

\begin{lemma}
  \label{lem:get-class-assignment}
  Let $p$ be a prime, $n$ be a positive integer, and let $k$ and $q$
  be integers with $k\geq q$.
  There is a randomized algorithm that takes
  a simple instance $S$ of $\twolinz{p^n}$ 
  and integers $k$ and $q$ as input, and in $O^*(2^{O(k)})$ time
  returns a class assignment $\tau : V(S) \to \EQ_{p^n}$
  such that the following holds.
  Let $Y$ be the set of equations in $S$ violated by $\tau$.
  Then have $|Y| \leq 2q$.
  Moreover, if $S$ admits an assignment of cost at most $k$ that
  strongly violates at most $q$ equations,
  then with probability at least $2^{-O(k+q^2)}$,
  $S-Y$ admits an assignment of cost at most $k - |Y| / 2$
  that agrees with $\tau$.
\end{lemma}
We mention that the algorithm in the previous lemma
can be derandomized using standard methods~\cite{chitnis2015directed}.
While technical details are deferred to
Section~\ref{ssec:class-assignments},
we remark that the equivalence class $\{0\}$ 
plays a special role:
since all binary equations in $S$ are homogeneous,
they are satisfied by the all-zero assignment.
Intuitively, this observation
allows us to ``greedily''
assign the value $0$ to the variables,
with the only obstacle being the crisp
unary equations in $S$.
After phrasing the class assignment problem
in terms of cuts in a certain auxiliary graph,
we apply the shadow removal technique.

To explain how the absorbing property of $\equiv$ is used,
we need some definitions.
Choose an arbitrary representative element $\rep{C}$
from every equivalence class $C \in \EQ$.
Consider a simple equation $e$ and a class 
assignment $\tau: V(e) \to \EQ$ to its variables
that respects $e$.
For unary equations $e = (u = r)$,
define $e' = \nxt(e, \tau)$ to be
\[ 
  u' = \frac{r - \rep{\tau(u)}}{p}.
\] 
For binary equations $e = (u = rv)$,
define $e' = \nxt(e, \tau)$ to be
\[
  u' = rv' + \frac{r \rep{\tau(v)} - \rep{\tau(u)}}{p}.
\]
The absorbing property implies that
$\nxt(e, \tau)$ is defined in both cases.
Indeed, if $e = (u = r)$ and $\tau$ respects $e$,
then $r \in \tau(u)$ and $r \equiv \rep{\tau(u)}$,
so $p$ divides $r - \rep{\tau(u)}$.
If $e = (u = rv)$ and $\tau$ respects $e$,
then there is an assignment $\alpha : \{u,v\} \to \ZZ_{p^n}$
such that
$\alpha(u) = r \alpha(v)$,
$\alpha(u) \in \tau(u)$ and
$\alpha(v) \in \tau(v)$.
Equivalently, $\alpha(u) - r\alpha(v) = 0$ and
$p$ divides both $\rep{\tau(u)} - \alpha(u)$ and
$\rep{\tau(v)} - \alpha(v)$,
hence $p$ also divides any linear combination 
of these two values, particularly
\[r(\rep{\tau(v)} - \alpha(v)) - (\rep{\tau(u) - \alpha(u)}) = 
r\rep{\tau(v)} - \rep{\tau(u)} + (\alpha(u) - r\alpha(v)) = 
r\rep{\tau(u)} - \rep{\tau(v)}.\]

\begin{lemma}
  \label{lem:next-level}
  Let $p$ be a prime and $n \in \ZZ_+$.
  Let $e$ be a simple equation over $\ZZ_{p^n}$,
  and $\tau : V(e) \to \EQ$ be a class assignment.
  Then $\tau$ respects $e$ 
  if and only if $\nxt(e, \tau)$ is satisfiable
  over $\ZZ_{p^{n-1}}$.
\end{lemma}
\begin{proof}
  If an assignment $\alpha : V(e) \to \ZZ_{p^n}$ satisfies $e$ 
  and agrees with $\tau$, define $\alpha' : V(e') \to \ZZ_{p^{n-1}}$
  as $\alpha'(x') = \frac{\alpha(x) - \rep{\tau(x)}}{p}$ for all $v \in V(e)$.
  Note that $\alpha'$ is well defined because $\equiv$ is absorbing.
  It is straightforward to verify that $\alpha'$ satisfies $e'$.
  
  On the other hand, if $\beta' : V(e') \to \ZZ_{p^{n-1}}$
  satisfies $e'$, then define
  $\beta : V(e) \to \ZZ_{p^n}$ as 
  $\beta(v) = \rep{\tau(v)} + p \cdot \beta'(v)$ for all $v \in V(e)$.
  Then $\beta$ agrees with $\tau$ 
  and satisfies $e$.
\end{proof}

Now we combine all ingredients to prove the main theorem.

\theoremalgo*
\begin{proof}
  Let $m = p_1^{n_1} \cdots p_\ell^{n_\ell}$ be the prime factorization of $m$.
  Note that $\omega(m) = \ell$, so
  by Proposition~\ref{prop:product-approx},
  it suffices to show that 
  $\mintwolinz{p^{n}}$ is \FPT-approximable within factor $2$
  for every prime $p$ and positive integer $n$.
  We proceed by induction on $n$.
  If $n = 1$, then we can use the algorithm of~\cite{Dabrowski:etal:soda2023}
  for $\mintwolin$ over fields to solve the problem exactly.
  Otherwise, let $(S, k)$ be an instance of $\mintwolinz{p^n}$.
  By Lemma~\ref{lem:simple-instances}, 
  we may assume without loss of generality that $S$ is simple.

  Select a value $q \in \{0,\ldots,k\}$ uniformly at random and
  run the algorithm from Lemma~\ref{lem:get-class-assignment} 
  on $(S, k, q)$ to produce a class assignment 
  $\tau : V(S) \to \EQ_{p^n}$.
  Let $Y$ be the set of equations in $S$ violated by $\tau$.
  If $|Y| > 2q$, then reject $(S, k)$.
  Otherwise, create an instance $S'$ of $\lin{2}{\ZZ_{p^{n-1}}}$
  with $V(S') = \{v' : v \in V(S)\}$ and
  $S' = \{\nxt(e, \tau) : e \in S - Y\}$.
  Set $k' = k - \lceil |Y|/2 \rceil$ and
  pass $(S', k')$ as input to the algorithm for 
  $\mintwolinz{p^{n-1}}$,
  and return the same answer.

  Towards correctness, first assume $(S, k)$ is a yes-instance,
  and let $q$ be the number of equations strongly violated 
  by an optimal assignment to $S$.
  Our guess for $q$ is correct with probability $\frac{1}{k+1}$.
  By Lemma~\ref{lem:get-class-assignment}, 
  with probability $2^{-O(q^2)}$,
  we have $|Y| \leq 2q$ and
  $S - Y$ admits an assignment of 
  cost at most $k - |Y|/2$ that agrees with $\tau$.
  By Lemma~\ref{lem:next-level},
  the instance $S'$ of $\twolinz{p^{n-1}}$ we create 
  has cost at most $k'$, so we accept.

  For the other direction, suppose the algorithm accepts.
  Then $(S', k')$ is a yes-instance of $\mintwolinz{p^{n-1}}$,
  and by Lemma~\ref{lem:next-level},
  $\mincost(S - Y) \leq k - |Y|/2$.
  Let $Z$ be an optimal solution to $S - Y$.  
  Then $Y \cup Z$ is a solution to $S$,
  and $|Y \cup Z| \leq |Y| + |Z| \leq 2(k - |Y|/2) + |Y| \leq 2k$.

  We remark that all steps in the algorithm
  can be derandomized either by exhaustive branching
  or standard methods in~\cite{marx2014fixed}.
\end{proof}

\subsection{Iterative Compression and Homogenization via Branching}
\label{ssec:comp-hom-branch}

Recall that an instance of $\mintwolinz{m}$ is simple
if every equation is crisp, unary and of the form
$u = r$ or binary and of the form $u = rv$,
where $u,v$ are variables and $r \in \ZZ_{m}$.
Using iterative compression and a trick called
homogenization, we will prove the following result,
which allows us to restrict our attention to simple instances.

\lemsimple*
\begin{proof}
  Suppose $\mintwolinz{m}$ is \FPT-approximable within a factor of $c$
  when restricted to simple instances.
  Start with an arbitrary instance $(S, k)$ of the problem.
  We may assume without loss of generality that 
  all unary equations in $S$ are crisp:
  indeed, if this is not the case, we may
  introduce a new variable $w$ in $V(S)$,
  add a crisp equation $w = 0$
  and replace every soft unary equation
  of the form $ax = b$
  with a binary equation $ax = w + b.$
  By iterative compression, assume we have
  access to a set $X \subseteq S$ of size at most $ck+1$
  such that $S - X$ is satisfiable.
  Fix an assignment $\chi$ that satisfies $S - X$.
  Construct a new instance $S'$ of $\twolinz{m}$ as follows.
  Start by adding copies of all variables from $V(S)$ to $V(S')$.
  For every soft/crisp equation $e = (au = bv + c)$,
  where $a,b,c \in \ZZ_{m}$, $a \neq 0$ and $u,v \in V(S)$,
  introduce a new variable $z_e$ in $S'$, and
  add crisp/soft equations $z_e = a u$ and $z_e = b v$.
  Now, for every possible assignment 
  $\alpha : V(X) \to \ZZ_{m}$,
  create an instance $S_{\alpha}$ starting with $S'$
  and adding crisp equations $v = \alpha(v)$ for all $v \in X$.
  Note that every instance $S_{\alpha}$ is simple.
  Thus, we can run the algorithm for $\mintwolinz{m}$
  on simple instances $(S_{\alpha}, k)$ for every $\alpha$,
  accept if any one of them is accepted,
  and reject otherwise.
  Note that the construction of each $S_{\alpha}$ takes
  polynomial time, and there are $2^{O(k)}$ of them,
  so the algorithm runs in \FPT-time.

  To prove correctness, it suffices to show that
  $\mincost(S) = \min_{\alpha} \mincost(S_{\alpha})$.
  First, let $\psi : V(S) \to \ZZ_{p^n}$ be 
  an arbitrary assignment to $S$,
  and fix $\alpha$ such that
  $\alpha(x) = \psi(x) - \chi(x)$ for all $x \in V(X)$.
  We claim that $\mincost(S_{\alpha}) \leq \cost(\psi)$.
  Indeed, the assignment defined as 
  \begin{itemize}
    \item $v \mapsto \psi(v) - \chi(v)$ for all $v \in V(S)$, and
    \item $z_e \mapsto a (\psi(u) - \chi(u))$ for all $e \in S - Z$
  \end{itemize}
  satisfies all equations in $S_{\alpha} \setminus S'$.
  For every $e = (au = bv + c)$ in $S$,
  it satisfies the equation $z_e = au$ in $S'$ by 
  definition; moreover, if
  $\psi$ satisfies $e$, then the assignment 
  also satisfies $z_e = bv$ because
  $a\psi(u) = b\psi(v) + c$ and
  $a\chi(u) = b\chi(v) + c$,
  so taking their difference yields
  $a(\psi(u) - \chi(u)) = b(\psi(v) - \psi(u))$.

  On the other hand, suppose $\phi$ is an assignment to $S_{\alpha}$
  such that $\phi$ agrees with $\alpha$ on $V(X)$.
  We claim that $\mincost(S) \leq \cost(\psi')$.
  Indeed, the assignment $v \mapsto \phi(v) + \chi(v)$
  satisfies every $e = (au = bv + c)$ in $S$ whenever
  $\phi$ satisfies both $z_e = au$ and $z_e = bv$
  because summing up
  $a\phi(u) = b\phi(v)$ with
  $a\chi(u) = b\chi(v) + c$ yields
  $a(\phi(u) + \chi(u)) = b(\phi(v) + \chi(v)) + c$.
\end{proof}

\subsection{Computing Class Assignments}
\label{ssec:class-assignments}

\newcommand{\RZ}{\ZZ_{p^n}}
\newcommand{\homoI}{simple}

This subsection is devoted to a proof of~\Cref{lem:get-class-assignment}. 
To achieve this we introduce the class assignment graph (Section~\ref{sec:class-graph}) and show that
certain cuts in this graph correspond to class assignments (Section~\ref{sec:class-cuts}), which
themselves correspond to solutions of $\minlin{2}{\RZ}$. We then
use shadow removal and branching to compute these cuts
(Section~\ref{sec:goodbye-shadows}). Throughout this section, we
consider the ring $\RZ$ for prime number $p$ and integer $n$.

\subsubsection{The Class Assignment Graph}
\label{sec:class-graph}

\newcommand{\EQM}[2]{\pi_{#1}}
\newcommand{\mundef}{\textsf{Nil}}

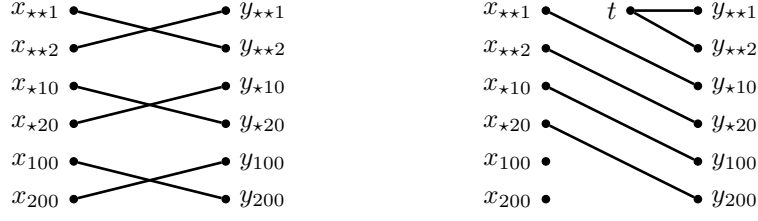
\begin{figure}
 \centering
 \begin{tikzpicture}[node distance=0.5cm]
    \tikzstyle{no}=[draw,circle, inner sep=1pt, fill]
    \tikzstyle{ed}=[draw,color=black, line width=1pt]

    \draw
    node[no, label=left:$x_{\star\star1}$] (l1) {}
    node[no, below of=l1, label=left:$x_{\star\star2}$] (l2) {}
    node[no, below of=l2, label=left:$x_{\star10}$] (l3) {}
    node[no, below of=l3, label=left:$x_{\star20}$] (l4) {}
    node[no, below of=l4, label=left:$x_{100}$] (l5) {}
    node[no, below of=l5, label=left:$x_{200}$] (l6) {}
    ;
    \draw
    (l1) +(2cm,0cm) node[no, label=right:$y_{\star\star1}$] (r1) {}
    node[no, below of=r1, label=right:$y_{\star\star2}$] (r2) {}
    node[no, below of=r2, label=right:$y_{\star10}$] (r3) {}
    node[no, below of=r3, label=right:$y_{\star20}$] (r4) {}
    node[no, below of=r4, label=right:$y_{100}$] (r5) {}
    node[no, below of=r5, label=right:$y_{200}$] (r6) {}
    ;

    \draw
    (l1) edge[ed] (r2)
    (l2) edge[ed] (r1)
    (l3) edge[ed] (r4)
    (l4) edge[ed] (r3)
    (l5) edge[ed] (r6)
    (l6) edge[ed] (r5)
    ;
\end{tikzpicture}
 \hspace{2cm}
 \begin{tikzpicture}[node distance=0.5cm]
    \tikzstyle{no}=[draw,circle, inner sep=1pt, fill]
    \tikzstyle{ed}=[draw,color=black, line width=1pt]

    \draw
    node[no, label=left:$x_{\star\star1}$] (l1) {}
    node[no, below of=l1, label=left:$x_{\star\star2}$] (l2) {}
    node[no, below of=l2, label=left:$x_{\star10}$] (l3) {}
    node[no, below of=l3, label=left:$x_{\star20}$] (l4) {}
    node[no, below of=l4, label=left:$x_{100}$] (l5) {}
    node[no, below of=l5, label=left:$x_{200}$] (l6) {}
    node[no, right=1cm of l1, label=left:$t$] (l0) {}
    ;
    \draw
    (l1) +(2cm,0cm) node[no, label=right:$y_{\star\star1}$] (r1) {}
    node[no, below of=r1, label=right:$y_{\star\star2}$] (r2) {}
    node[no, below of=r2, label=right:$y_{\star10}$] (r3) {}
    node[no, below of=r3, label=right:$y_{\star20}$] (r4) {}
    node[no, below of=r4, label=right:$y_{100}$] (r5) {}
    node[no, below of=r5, label=right:$y_{200}$] (r6) {}
    ;

    \draw
    (l0) edge[ed] (r1)
    (l0) edge[ed] (r2)
    (l1) edge[ed] (r3)
    (l2) edge[ed] (r4)
    (l3) edge[ed] (r5)
    (l4) edge[ed] (r6)
    ;
\end{tikzpicture}
 \caption{Edges in the conformal graph from $\ZZ_{27}$
 corresponding to the equations $x=2y$ and $x=3y$, respectively.
 The notation for equivalence classes is inspired by base-$3$
 encoding of the numbers in $\{0,\dots,26\}$, namely
  ${\star}{\star}1 = \{001_3, 011_3, 021_3, 101_3, \dots, 221_3\}$,
  ${\star}{\star}2 = \{001_3, 011_3, 021_3, 101_3, \dots, 222_3\}$,
  ${\star}10 = \{010_3, 110_3, 210_3\}$,
  ${\star}20 = \{020_3, 120_3, 220_3\}$,
  $100 = \{100_3\}$,
  $200 = \{200_3\}$.
 }
 \label{fig:z27-matchings}
\end{figure}

In what follows, let $(S,k)$ be a \homoI{} instance of
$\minlin{2}{\RZ}$ and let $\equiv$ be the matching and absorbing
equivalence relation on $\RZ$ defined in
\Cref{ssec:algorithm-summary}.
Without loss of generality, we assume that every equation of $S$ is
consistent (otherwise we can remove the equation and decrease $k$ by
$1$)
and non-trivial (otherwise we can remove the equation).
We first use the matching property of $\equiv$ to define the mapping 
$\EQM{e}{i} : \EQ_{p^n}^{\neq 0} \rightarrow \EQ_{p^n}$
between equivalence classes for any equation $e=(ax=y)$ with $a
\in \RZ\setminus\{0\}$, as follows.
For every $C \in \EQ_{p^n}$, we set:
\begin{itemize}
\item $\EQM{e}{i}(C)=0$ if $ar=0$ for every $r \in C$ and 
\item $\EQM{e}{i}(C)=D$ otherwise, where $D$
  is the unique equivalence class such that $e$ maps $C$ to $D$, which
  is uniquely defined because $\equiv$ has the matching property.
\end{itemize}
Moreover, for an equation $e=(x=a)$ with $a \in \RZ$,
we let $\EQM{e}{i}$ be the unique equivalence class $\EQ_{p^n}(a)$ consistent with $e$.

\begin{figure}
 \centering
  \begin{tikzpicture}[node distance=0.5cm]
    \tikzstyle{no}=[draw,circle, inner sep=1pt, fill]
    \tikzstyle{edno}=[inner sep=3pt]
    \tikzstyle{ed}=[draw,color=black, line width=1pt]

    \draw[xscale=2,yscale=1] 
    node[no, label=left:$s$] (s) {}
    (s) +(1cm,-1cm) node[no, label=below:$b_O$] (b1) {}
    (b1) +(1cm,0cm) node[no, label=below:$d_E$] (d2) {}
    (d2) +(1cm,0cm) node[no, label=below:$r_E$] (r2) {}
    (r2) +(1cm,0cm) node[no, label=below:$u_E$] (u2) {}
    (u2) +(1cm,0cm) node[no, label=below:$d_O$] (d1) {}
    (s) +(1cm,1cm) node[no, label=above:$a_O$] (a1) {}
    (a1) +(1cm,0cm) node[no, label=above:$c_E$] (c2) {}
    (c2) +(1cm,0cm) node[no, label=above:$u_O$] (u1) {}
    (u1) +(1cm,0cm) node[no, label=above:$r_O$] (r1) {}
    (r1) +(1cm,0cm) node[no, label=above:$c_O$] (c1) {}

    (c1) +(1cm,-1cm) node[no, label=right:$t$] (t) {}
    ;

    \draw
    (s) -- node[edno, midway, rotate=-26, anchor=north] (sb1) {$b=1$} (b1)
    (b1) -- node[edno, midway, anchor=north] (db) {$2b=d$} (d2)
    (d2) -- node[edno, midway, anchor=north] (rd) {$d=r$} (r2)
    (r2) -- node[edno, midway, anchor=north] (ur) {$r=u$} (u2)
    ;

    \draw
    (s) -- node[edno, midway, rotate=26, anchor=south] (sa1) {$a=1$} (a1)
    (a1) -- node[edno, midway, anchor=south] (ac) {$2a=c$} (c2)
    (c2) -- node[edno, midway, anchor=south] (cu) {$c=2u$} (u1)
    (u1) -- node[edno, midway, anchor=south] (ur) {$u=r$} (r1)
    ;

    \draw
    (t) -- node[edno, midway, rotate=-26, anchor=north] (c1t1)
    {$c=2a$} (c1)
    (t) -- node[edno, midway, rotate=-26, anchor=south] (c1t2) {$c=2u$} (c1)
    ;
    \draw
    (t) -- node[edno, midway, rotate=26, anchor=north] (d1t1)
    {$d=2b$} (d1)
    ;
    \draw (r1) -- node[edno, midway, rotate=-45, anchor=south] (r1d1) {$r=d$} (d1);

  \end{tikzpicture}
 \caption{Let $S$ be the instance of $\lin{2}{\ZZ_4}$ with variables
   $a$, $b$, $c$, $d$, $u$, $r$ and equations $a=1$, $b=1$, $2a=c$,
   $c=2u$, $u=r$, $2b=d$, and $d=r$. The figure illustrate the class
   assignment graph $G=G(S)$. Note that $\equiv$
   has only two non-zero equivalence classes, namely, $E=\{2\}$ and
   $O=\{1,3\}$. Every edge of $G$ is annotated with the equation that
   implies it. A minimum conformal $st$-cut is given by the two edges
   that correspond to the equation $u=r$ and corresponds to the class
   assignment $a=O$, $b=O$, $c=E$, $d=E$, $u=O$, and
   $r=E$. Note that $G$ has only one minimal conformal $st$-cut
   closest to $s$, namely $\{a_Oc_E, b_Od_E\}$. This $st$-cut
   corresponds to the class assignment $a=O$, $b=O$, and
   $c=d=u=r=0$. Therefore, the optimum solution for $S$ only removes
   the equation $u=r$, however, any solution that corresponds to a
   minimum conformal $st$-cut closest to $s$ has to remove the
   equations $2a=c$ and $2b=d$.}
 \label{fig:classassignment}
\end{figure}
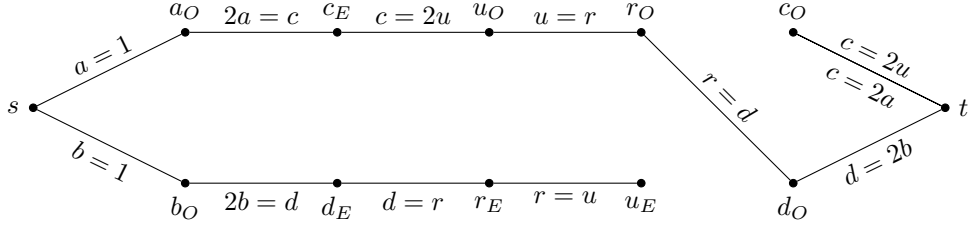

We are now ready to define the \emph{class assignment} graph
$G=G(S)$; see Figure~\ref{fig:classassignment} for an illustration
of such a graph. The graph $G$ has two distinguished vertices $s$ and $t$
together with vertices $x_{C}$ for every $x \in V(S)$ and every non-zero
class $C \in \EQ_{p^n}^{\neq 0}$. Moreover, $G$ contains the following edges for each equation.
\begin{itemize}
\item For an equation $e=(ax=y)$ we do the following. For every $C
  \in \EQ_{p^n}^{\neq 0}$, we add the edge
  $x_Cy_{\EQM{e}{i}(C)}$ if $\EQM{e}{i}(C) \neq 0$.
  For
  every $D \in \EQ_{p^n}^{\neq 0}$ such that $\EQM{e}{i}^{-1}(D)$ is
  undefined, we add the edge $y_Dt$.
  If the equation is crisp, then the added edges are also crisp.
\item For a crisp equation of the form $e=(x = 0)$, we add the crisp
  edge $x_{C}t$ for every class $C \in \EQ_{p^n}^{\neq 0}$.
\item For a crisp equation of the form $e=(x = b)$, where $b\neq 0$ ,
  we add the crisp edge $sx_{\EQM{e}{i}}$.
\end{itemize}
See Figure~\ref{fig:z27-matchings} for an illustration.

Intuitively, every node of $G$ corresponds to a Boolean variable and
every edge $e$ of $G$ corresponds to an ``if and only if'' between the two
Boolean variables connected by $e$. Moreover, every
assignment $\varphi$ of the variables of $S$ naturally corresponds to a
Boolean assignment, denoted by $\BA{\varphi}$
of the vertices in $G$ by setting $s=1$, $t=0$, and:
\begin{itemize}
\item if $\varphi(x)$ belongs to the non-zero class $C$, then we set
  $x_{C}=1$ and $x_{C'}=0$ for every non-zero class $C'$ not
  equal to $C$,
\item if $\varphi(x)=0$, we set $x_{C}=0$ for every non-zero class $C$.
\end{itemize}
We say that an
edge $e$ of $G$ is satisfied by $\varphi$ if $\BA{\varphi}$ satisfies the
``if and only if'' Boolean constraint represented by that edge.
\begin{observation}\label{obs:ass-ba-ass}
  Let $S$ be a \homoI{} instance of $\lin{2}{\RZ}$, let $\varphi$
  be an assignment of $S$ and let $e$ be an equation in
  $S$. If $\varphi$ satisfies $e$, then $\BA{\varphi}$ satisfies all
  edges corresponding to $e$ in $G(S)$.
\end{observation}

\subsubsection{Cuts in the Class Assignment Graph}
\label{sec:class-cuts}

In this section, we introduce {\em conformal cuts} and show how they
relate to class
assignments and solutions to
$\minlin{2}{\RZ}$ instances. Let $S$ be a \homoI{} instance of
$\lin{2}{\RZ}$ and $G=G(S)$.
An $st$-cut $Y$ in $G$ is \emph{conformal} if
for every variable $x \in V(S)$ at most one vertex $x_{C}$ for some $C
\in \EQ_{p^n}^{\neq 0}$ is
connected to $s$ in $G-Y$. Please refer to
Figure~\ref{fig:classassignment} for an illustration of conformal cuts
in the class assignment graph.
If $Y$ is a conformal $st$-cut in $G$, then we say that a variable $x$
is \emph{decided} with respect to $Y$ if (exactly) one vertex $x_{C}$ is reachable from $s$ in
$G-Y$; and otherwise we say that $x$ is \emph{undecided} with respect to $Y$. Moreover, we denote by $\clasn{Y}$ the assignment of variables
of $S$ to classes in $\EQ_{p^n}$ implied by $Y$, i.e.\ $\clasn{Y}(x)=0$ if
$x$ is undecided and otherwise $\clasn{Y}(x)=C$, where $C$
is the unique non-zero class in $\EQ_{p^n}$ such that $x_C$ is
reachable from $s$ in $G-Y$.
We say that an assignment $\varphi$ of $S$ \emph{agrees with} $Y$
if $\varphi(x)$ is in the class $\clasn{Y}(x)$ for every variable $x$
of $S$.
Note that if some assignment agrees with $Y$, then $Y$ is
conformal. The following auxiliary lemma characterizes which edges of
$G$ are satisfied by an assignment $\varphi$ of $S$ after removing a
set $Y$ of edges from $G$.

\begin{lemma}\label{lem:ze-ba-ass}
  Let $Y$ be a set of edges of $G$ and let $\varphi$ be an assignment
  of $S$. Then, $\BA{\varphi}$ satisfies all edges 
  reachable from $s$ in $G-Y$ if and
  only if $\BA{\varphi}$ sets all Boolean variables reachable from $s$ in
  $G-Y$ to $1$. Similarly, $\BA{\varphi}$ satisfies all edges
  reachable from $t$ in $G-Y$ if and
  only if $\BA{\varphi}$ sets all Boolean variables reachable from $t$ in
  $G-Y$ to $0$.
\end{lemma}
\begin{proof}
  This follows because $\BA{\varphi}(s)=1$ and $\BA{\varphi}(t)=0$ for any $\varphi$
  and every edge of $G$ corresponds to an ``if and only if''
  between the variables corresponding to its two endpoints.
\end{proof}

For a set $Z$ of equations of $S$, we let $\ed(Z)$ denote the set of
all edges of $G$ corresponding to an equation of $Z$.
Conversely, for a set $Y$
of edges of $G$, we let $\con(Y)$ denote
all equations of $S$ having a corresponding edge in $Y$. Moreover, if
$Y$ is an $st$-cut in $G$, we let $\sep(Y)$ denote the unique
minimal $st$-cut contained in $Y$ that is \emph{closest} to $s$ in $G$.
Finally, for an optimal
solution $Z$ of $S$, we let $\comp{Z}$
be the set of equations
$Z\setminus \con(\sep(\ed(Z)))$, i.e.\ all equations in $Z$ that do
not have an edge in $\sep(\ed(Z))$.

We are now ready to show the main result of this subsection, which,
informally, states that every set $Z$ of equations for which $S-Z$ is
satisfiable gives rise to a
conformal $st$-cut of size at most $2 \cdot |Z\setminus \comp{Z}|$ that
agrees with some satisfying assignment of $S-Z$.
Thus, we establish
the required connection between solutions of $\minlin{2}{\RZ}$
instances and conformal
$st$-cuts in the class assignment graph.

\begin{lemma} \label{lem:deleted-edges}
  Let $Z$ be a set of equations such that $S-Z$ is satisfiable and let $Y=\ed(Z)$.
  Then,
  $Y'=\sep(Y)$ satisfies:
  \begin{enumerate}
  \item $Y'$ is a conformal $st$-cut.
  \item $|Y'|\leq 2|\con(Y')| = 2|Z \setminus \comp{Z}|$.
  \item There is a satisfying assignment for $S-Z$ that agrees with $Y'$.
  \end{enumerate}
\end{lemma}
\begin{proof}
  Let $\varphi$ be a satisfying assignment of $S-Z$.
  \Cref{obs:ass-ba-ass} implies that $\BA{\varphi}$ satisfies
  all edges of $G-Y$. Therefore, it follows from
  Lemma~\ref{lem:ze-ba-ass} that $\BA{\varphi}$ sets all vertices
  reachable from $s$ in $G-Y$ to $1$ and all vertices reachable from $t$ in $G-Y$ to $0$.
  Thus, $Y$ is an $st$-cut, because otherwise $t$ would
  have to be set to $1$ by $\BA{\varphi}$ since it would be reachable from
  $s$ in $G-Y$. Therefore, $Y'=\sep(Y)$ exists. Because $Y'$ is
  closest to $s$, it holds that a vertex is reachable from $s$ in
  $G-Y$ if and only if it is reachable from $s$ in $G-Y'$.
  Therefore, if at least two vertices $x_{C}$ and $x_{C'}$ for some
  distinct non-zero classes $C$ and $C'$ are
  reachable from $s$ in $G-Y'$ for some variable $x$, then all of them must be set to $1$ by
  $\BA{\varphi}$, which is not possible due to the definition of
  $\BA{\varphi}$. We conclude that $Y'$ is conformal.

  Towards showing that $|Y'|\leq 2|\con(Y')|$, it suffices to show that $|Y'\cap
  \ed(e)|\leq 2$ for every equation $e \in Z$. Note that because $Y'$
  is a minimal $st$-cut, it holds that one of the endpoints of every
  $y \in Y'$ is reachable from $s$ in $G- Y'$. Therefore, because $Y'$
  is conformal, $Y'$ can contain at most two edges in $\ed(e)=\SB
  x_Cy_{\EQM{e}{i}(C)} \SM C \in \EQ_{p^n}^{\neq 0}\land \EQM{e}{i}(C)\neq
  0\SE\cup \SB y_Dt\SM D \in \EQ_{p^n}^{\neq 0} \land
  \EQM{e}{i}^{-1}(D)\textup{ is undefined}\SE$ for every binary equation $e$ of the
  form $ax=y$. Similarly, $Y'$ can contain at most one edge in $\ed(e)=\SB
  x_Ct \SM C \in \EQ_{p^n}^{\neq 0}\SE$ for every unary equation $e$ of the
  form $x=0$. Finally, $|\ed(e)|=|\{sx_{\EQM{e}{i}}\}|=1$ for every
  unary equation $e$ of the form $x=b$. 
  Therefore, $|Y'| \leq 2|\con(Y')|$, and $Y'=Z \setminus \comp{Z}$ by definition.

  Let $D$ be the set
  of all variables of $S$ such that no vertex $x_{C}$ is
  reachable from $s$ in $G-Y'$. Let $\varphi'$ be the assignment for $S$ such that $\varphi'(x)=0$ if
  $x \in D$ and $\varphi'(x)=\varphi(x)$ otherwise. Clearly, $\varphi'$ agrees
  with $Y'$, because $\varphi$ agrees with all variables not in $D$ and
  all other variables are correctly set to $0$ by $\varphi'$. It
  therefore only remains to show that $\varphi'$
  still satisfies $S-Z$.

  Consider a unary equation $e$ of $S-Z$ on a variable $x$.
  If $x
  \notin D$, then $\varphi'(x)=\varphi(x)$ and $\varphi'$ satisfies
  $e$. So suppose that $x \in D$. If $e$ is of the form $x=0$, then
  $\varphi'$ satisfies $e$. Otherwise, $e$ is of the form $x=b$
  for some $b\neq 0$ and $G-Y'$ contains the edge
  $sx_{\EQM{e}{i}}$. Therefore, $x_{\EQM{e}{i}}$ is reachable from $s$ in
  $G-Y'$ contradicting our assumption that $x \in D$.

  Consider instead a binary equation $e$ of $S-Z$ on
  two variables $x$ and $y$. If $x,y \in D$, then $e$ is clearly
  satisfied by $\varphi'$. Similarly, if $x,y \notin D$, then $e$ is
  also satisfied by $\varphi'$, because $\varphi'(x)=\varphi(x)$ and
  $\varphi'(y)=\varphi(y)$ and $\varphi$ satisfies $e$. So suppose that
  $|\{x,y\}\cap D|=1$ and that $e$ is of the form $ax =y$ for some
  $a\in \RZ$.
  Note that $\ed(e) \subseteq E(G)\setminus
  \ed(Z)$. Suppose first that $x \notin D$ and
  $y \in D$. Because $x \notin D$, there is a vertex
  $x_{C}$ reachable from $s$ in $G-Y'$. If $\EQM{e}{i}(C) \neq 0$,
  then because $x_{C}y_{\EQM{e}{i}(C)} \notin
  Y'$, the vertex $y_{\EQM{e}{i}(C)}$ is reachable from $s$ in $G-Y'$ contradicting our assumption
  that $y \in D$. Otherwise, $\EQM{e}{i}(C)=0$, which together with our assumption that
  $\varphi'(y)=0$ shows that $e$ is satisfied by $\varphi'$.
  
  Suppose now that $x \in D$ and
  $y \notin D$. Because $y \notin D$, there is a vertex
  $y_{D}$ reachable from $s$ in $G-Y'$. Because $e$ is satisfied by $\varphi$, it holds that
  $\EQM{e}{i}^{-1}(D)$ is defined. But then, because $x_{\EQM{e}{i}^{-1}(D)}y_{D} \notin
  Y'$, the vertex $x_{\EQM{e}{i}^{-1}(D)}$ is reachable from $s$ in $G-Y'$ and this contradicts our assumption
  that $x \in D$. 
\end{proof}

\subsubsection{Shadow Removal}
\label{sec:goodbye-shadows}

We continue by showing how the shadow removal technique
(introduced in~\cite{marx2014fixed} and improved
in~\cite{chitnis2015directed}) can be used for computing conformal cuts that correspond to solutions of
a $\minlin{2}{\RZ}$ instance.
We follow~\cite{chitnis2015directed} and begin by importing some definitions,
which we translate from directed graphs to undirected graphs to fit
our setting; to get back to directed graphs one simply has to think of
an undirected graph as the directed graph obtained after replacing each
undirected edge with two directed arcs in both directions.
Let $G$ be an undirected graph.
Let $\cF$ be a set of connected subgraphs of $G$.
A set $T \subseteq V(G)$ is an \emph{$\cF$-transversal} 
if $T$ intersects every subgraph in $\cF$.
Conversely, if $T$ is an $\cF$-transversal,
we say that $\cF$ is \emph{$T$-connected}.

\begin{theorem}[{\cite[Theorem 3.5]{chitnis2015directed}}]
  \label{thm:shadow-cover}
  Let $G$ be an undirected graph, $T \subseteq V(G)$ and $k \in \naturals$.
  There is a randomized algorithm that takes $(G, T, k)$
  as input
  and returns in $O^*(4^k)$ time a set $W \subseteq V(G) \setminus T$
  such that the following holds with probability $2^{-O(k^2)}$.
  For every $T$-connected family of connected subgraphs $\cF$ in $G$,
  if there is an $\cF$-transversal of size at most $k$ in $V(G) \setminus T$,
  then there is an $\cF$-transversal $Y \subseteq V(G) \setminus (W \cup T)$ 
  of size at most $k$ such that every vertex $v \notin W \cup Y$
  is connected to $T$ in $G - Y$.
\end{theorem}

To explain the name, say that a vertex $v \in V(G)$ is \emph{in shadow of $Y$} if $v$ is not
connected to $T$ in $G-Y$. Thus Theorem~\ref{thm:shadow-cover} returns a
\emph{shadow-covering set} that covers the shadow of $Y$ for some $\cF$-transversal $Y$
without covering $Y$ itself, i.e.\ the set $W$ in Theorem~\ref{thm:shadow-cover}. 

Theorem~\ref{thm:shadow-cover} builds on a notion of greedy cuts,
known as \emph{important separators} (see Marx~\cite{Marx2006parameterized}).  
Informally, an important $(S,T)$-separator is an $(S,T)$-cut $X$ that is
as close to $T$ as possible given its size $|X|$. The important separators
in a graph have significant structure, that is frequently useful for
FPT graph separation problems~\cite{Marx2006parameterized,marx2014fixed,cygan2015book}.

In Theorem~\ref{thm:shadow-cover}, let $Y$ be a minimal $\cF$-transversal with $|Y| \leq k$, and
let $Y'$ be the result of pushing $Y$ to be as close to $T$ as possible
while maintaining $|Y'| \leq |Y|$. Then $Y'$ is also an $\cF$-transversal,
and the greedy nature of $Y'$ allows the shadow of $Y'$ to be sampled
using structural properties of the important separators in $G$; see~\cite{chitnis2015directed}.

This notion of greedy cuts is also relevant to $\minlin{2}{\ZZ_{p^n}}$.
Let $S$ be a simple instance with class assignment graph $G$,
$Z \subseteq S$ a solution, $Y'=\sep(\ed(Z))$, and
let $D$ be the set of vertices reachable from $s$ in $G-Y'$.
If we somehow knew the cut $Y'$, then we could assume that 
all variables not represented in $D$ take the value $0$ in a
satisfying assignment of $S-Z$, while variables represented in $D$
need to be carried forward to the next instance (since we only know
their class assignment, not their exact value). 
Thus, if we push $Y'$ to a cut $Y''$ that cuts as close to $s$ as possible,
with $|Y''| \leq |Y'|$, then intuitively $Y''$ is a better cut than $Y'$
as it eliminates more variables.
$Y''$ may represent more equations than $Y'$, if $Y'$ contains
more than one edge per equation, but by Lemma~\ref{lem:deleted-edges} this
overhead is at most a factor of two.
Hence, shadow removal is applicable to $\minlin{2}{\ZZ_{p^n}}$,
up to getting an FPT 2-approximation.

The other aspect, of precisely how knowing the shadow-covering set
$W$ helps, is less obvious. 
To illustrate, let $S$ be a simple instance of $\minlin{2}{\ZZ_4}$, and
assume that $S$ contains an odd cycle $C$ consisting of equations $u=3v$,
together with equations that require a variable of the cycle to be odd.
Then the vertex $v_O$ representing $v \in \{1,3\}$ is connected to $s$ in the class assignment
graph $G$ for every $v \in V(C)$, and the system is infeasible.
There are now two classes of solutions: either $v$ takes an even value
for some $v \in V(C)$, or all variables $v \in V(C)$ take odd
values but some equation $u=3v$ is violated by letting $u=v$.
The former requires separating $v_O$ from $s$ by a conformal cut in $G$,
and the latter would be handled in the instance of $\minlin{2}{\ZZ_2}$
that we create in the next phase of the algorithm. Unfortunately,
if $C$ is large, there is in the former case an unbounded number of
vertices $v_O$ that we might want to separate from $s$.
The structure of the shadow-covering set $W$ simplifies this search,
reducing it to a bounded branching process. 

The following lemma formalizes the idea that cutting closer to $s$ is better.
Informally, it shows that if $Z$ is a solution, i.e.\ a set of equations such that $S-Z$ is satisfiable,
then we can obtain a (not too large) new solution $Z'=(\comp{Z} \cup \con(Y'))$ by replacing
the corresponding conformal minimal $sA$-cut $Y=\sep(\ed(Z))$, where $A$
is the set of vertices in $G$ not reachable from $s$ in $G-Y$, by any
minimal $sA$-cut $Y'$.
\begin{lemma}
  \label{lem:anti-sol}
  Let $S$ be a \homoI{} instance of $\lin{2}{\RZ}$ and $G=G(S)$.
  Moreover, let $Z$ be a set of equations such that $S-Z$ is satisfiable, 
  $Y=\sep(\ed(Z))$, $A$ be the set of all
  vertices in $G$ that are not reachable from $s$ in $G-Y$, and let
  $Y'$ be an $sA$-cut in $G$.
  Then, there is an assignment $\varphi : V(S) \rightarrow \RZ$ of $S$ that satisfies
  $S-Z'$ and agrees with $Y'$, where $Z'=(\comp{Z} \cup \con(Y'))$.
\end{lemma}
\begin{proof}
  Lemma~\ref{lem:deleted-edges} implies that $Y$ is conformal and there is a satisfying assignment
  $\varphi$ for $S-Z$ that agrees with $Y$.
  Because $Y'$ is also an $sA$-cut in
  $G$, if no vertex $x_{C}$ is reachable from $s$ in $G-Y$ for
  some variable $x$ of $S$, then
  the same applies in $G-Y'$. Let $D$ be the set of all variables
  $x$ of $S$ such that some vertex $x_{C}$ is reachable from $s$ in $G-Y$
  but that is not the case in $G-Y'$. Let $\varphi'$
  be the assignment obtained from $\varphi$ by setting all variables
  in $D$ to $0$. Then, $\varphi'$ agrees with $Y'$.
  We
  claim that $\varphi'$ also satisfies $S-Z'$, where $Z'=(\comp{Z} \cup
  \con(Y'))$.

  Consider a unary equation $e$ of $S-Z'$ on variable $x$.
  If $x
  \notin D$, then $\varphi'(x)=\varphi(x)$ and therefore $\varphi'$ satisfies
  $e$ (because $e$ is crisp and therefore $e \notin Z$). So suppose that $x \in D$. If $e$ is of the form $x=0$, then
  $\varphi'$ satisfies $e$. Otherwise, $e$ is of the form $x=b$
  for some $b\neq 0$ and $G-Y'$ contains the edge
  $sx_{\EQM{e}{i}}$. Therefore, $x_{\EQM{e}{i}}$ is reachable from $s$ in
  $G-Y'$ contradicting our assumption that $x \in D$.

  Now, consider a binary
  equation $e=(ax=y)$ of $S-Z'$ on variables $x$ and
  $y$ and first consider the case when $e \in Z$. Clearly, if neither
  a vertex $x_C$ nor a vertex $y_C$ is reachable from $s$ in $G-Y'$,
  then $\varphi'(x)=\varphi'(y)=0$, so $e$ is satisfied by
  $\varphi'$. We next show
  that either no vertex $x_C$ or no vertex $y_C$ is reachable from $s$
  in $S-Y'$. Suppose for a contradiction that $x_{C_x}$ and
  $y_{C_y}$ are reachable from $s$ in $S-Y'$.
  Let $h$ be an arbitrary edge in $\ed(e)\cap Y$; such an edge $h$ exists because $e \in Z\setminus Z'$.
  Because $Y$ is a minimal $st$-cut, it follows that exactly one
  endpoint of $h$ is reachable from $s$ in $G-Y$ and either $x_C$ or $y_C$ (endpoint of $h$) for some $C \in
  \EQ_{p^n}$ must be reachable from $s$ in $G-Y$. We assume
  without loss of generality that $x_C$ is reachable from $s$ in $G-Y$. Because $Y'$ is an
  $sA$-cut and $Y'$ does not contain $h$, $x_C$ is not reachable from $s$ in $G-Y'$. But then $C\neq C_x$
  and both $x_C$
  and $x_{C_x}$ are reachable from $s$ in $G-Y$, which contradicts that $Y$ is conformal.
  It remains to consider the case when there is a vertex $x_C$ that is reachable
  from $s$ in $G-Y'$ but no vertex $y_C$ is reachable from $s$ in
  $G-Y'$; the case when there is a vertex $y_C$ reachable from $s$ in
  $G-Y'$ but no vertex $x_C$ reachable from $s$ in $G-Y'$ is analogous.
  Since $Y'\cap \ed(e)=\emptyset$, we obtain that $\EQM{e}{i}(C)=0$
  since otherwise either $t$ or some $y_{C'}$ would be reachable from
  $s$ in $G-Y'$. Because $\varphi'(y)=0$, it follows that $e$ is
  satisfied by $\varphi'$. This completes the proof for the case when $e
  \in Z$.

  Suppose instead that $e \notin Z$. In this case $\varphi$ satisfies $e$ and
  therefore $\varphi'$ also satisfies $e$ unless exactly one of $x$ and
  $y$ is not in $D$. We distinguish the following cases:
  \begin{itemize}
  \item $x \notin D$ and $y \in D$. If there is no vertex $x_C$ that is
    reachable from $s$ in $G-Y'$, then the same holds in $G-Y$ so
    $\varphi'(x)=\varphi(x)=\varphi'(y)=0$, which shows that $\varphi'$
    satisfies $e$. Otherwise, let $x_C$ be reachable from $s$ in
    $G-Y'$. Then, $\EQM{e}{i}(C)=0$ since otherwise either $t$ or
    $y_{\EQM{e}{i}(C)}$ is also reachable from $s$ in $G-Y'$ (because
    $Y'\cap \ed(e)=\emptyset$), which in
    the former case contradicts our assumption that $Y'$ is an
    $st$-cut and which in the latter case contradicts our assumption
    that $y \in D$. Therefore, $\varphi'$ satisfies $e$ (because
    $\varphi'(y)=0$).
  \item $x \in D$ and $y \notin D$. We first show that there is no
    vertex $y_C$ that is reachable from $s$ in $G-Y'$. Suppose there
    is such a vertex $y_C$. Then, $\EQM{e}{i}^{-1}(C)$ is undefined
    since otherwise $x_{\EQM{e}{i}^{-1}(C)}$ is reachable from $s$ in
    $G-Y'$ (because $Y'\cap \ed(e)=\emptyset$), which contradicts our
    assumption that $x \in D$. But then $y_Ct \in E(G-Y')$ and
    $t$ is reachable from $s$ in $G-Y'$, which contradicts
    our assumption that $Y'$ is an $st$-cut. Hence, there is no vertex $y_C$ that is
    reachable from $s$ in $G-Y'$, which implies that the same holds in $G-Y$ so $\varphi'(y)=\varphi(y)=\varphi'(x)=0$, which shows that $\varphi'$
    satisfies $e$.\qedhere 
  \end{itemize}%
\end{proof}

Let $G=G(S)$ and for a set $W \subseteq
V(G)$, let $\delta(W)$ be the set of edges incident to a vertex in $W$
and a vertex in $V(G)\setminus W$.
The forthcoming \Cref{lem:our-shadow-cover} 
provides a version of shadow
removal adopted to our problem. 
Informally, it provides us with a set
$W \subseteq V(G)$ such that we only
have to look for conformal $st$-cuts that are subsets of $\delta(W)$ to
obtain our class assignment; in fact it even shows that for every component $C$ of $G[W]$ either all edges in $\delta(C)$ are part of the cut or no edge of $\delta(C)$ is part of the cut. We will use this fact in
\Cref{lem:list-conformal-cuts} to find a conformal $st$-cut by branching on
which components of $G[W]$ are reachable from $s$.

More formally, if $Z$ is a set of equations such that $S-Z$ is satisfiable and $A$ is the
set of vertices not reachable from $s$ in $G$ minus the conformal
$st$-cut $\sep(\ed(Z))$ (see Lemma~\ref{lem:deleted-edges}), then the
lemma provides us with a set $W \subseteq V(G)$ such that there is a
conformal $sA$-cut $Y'$ within $\delta(W)$ of size at most
$2|Z \setminus \comp{Z}|$ such that there is an assignment $\varphi
: V(S) \rightarrow \RZ$ for the variables in $S$ that satisfies
$S-(\comp{Z} \cup \con(Y'))$ and
agrees with $Y'$. The main idea behind the proof is the application of Theorem~\ref{thm:shadow-cover} to the set of all walks from $s$ to $A$ in $G$ to obtain the set $W$ and to employ Lemma~\ref{lem:anti-sol} to obtain the new solution that corresponds to the minimum $sA$-cut $Y' \subseteq \delta(W)$.

\begin{lemma}
  \label{lem:our-shadow-cover}
  Let $S$ be a \homoI{} instance of $\lin{2}{\RZ}$ and let $G=G(S)$.
  Moreover, let $Z$ be a set of equations such that $S-Z$ is satisfiable, 
  $Y=\sep(\ed(Z))$, let $A$ be the set of all
  vertices in $G$ that are not reachable from $s$ in $G-Y$, and let
  $q=|Z\setminus \comp{Z}|$.
  There is a randomized algorithm that in $\bigohs(4^{2q})$ time takes
  $(G,q)$ as input and returns  a set $W \subseteq V(G) \setminus \{s\}$
  such that the following holds with probability $2^{-O(q^2)}$.
  There is a (minimal) $sA$-cut $Y'$ of size at most $2q$
  satisfying:
  \begin{enumerate}
  \item every vertex $v \notin W$ is connected to $s$ in $G - Y'$,
  \item %
    $Y' \subseteq \delta(W)$,
  \item there is a satisfying assignment for $S-(\comp{Z} \cup \con(Y'))$
    that agrees with $Y'$
  \end{enumerate}
  Moreover, for every component $C$ of $G[W]$ the following holds:
  \begin{enumerate}[resume]
  \item either $Y'\cap\delta(C)=\emptyset$ or $\delta(C)\subseteq Y'$,
  \item if $t \in C$, then $\delta(C) \subseteq Y'$,
  \item if $x_{\alpha},x_{\alpha'} \in C$ for some variable $x$
    and $\alpha\neq \alpha'$, then $\delta(C)
    \subseteq Y'$,
  \item if $C$ contains some $x_\alpha$ for some decided
    variable $x$ w.r.t. $Y'$, then $\delta(C) \subseteq Y'$.
  \end{enumerate}
\end{lemma}
\begin{proof}
  In order to apply Theorem~\ref{thm:shadow-cover}, we first transform $G$ into
  an undirected graph $G'$ as follows. 
  For every vertex $a \in V(G)$, we create a clique on $2q+1$
  vertices $K(a) = \{a^1, \dots, a^{2q+1}\}$ in $G'$.
  For every edge $ab \in E(G)$, we introduce an auxiliary vertex $z_{ab}$ and connect it to $a^i$ and $b^i$
  for all $1 \leq i \leq 2q+1$. Informally, the
  construction of $G'$ is required to ensure
  that every $sA$-cut in $G$ corresponds to an equally sized
  transversal in $G'$.

  We now apply Theorem~\ref{thm:shadow-cover} to the tuple
  $(G',T,2q)$, where $T=K(s)$. Let $\cF$ be the set of all walks in $G'$ from
  $K(s)$ to $\bigcup_{a \in A}K(a)$; clearly $\cF$ is $T$-connected.

  We first show that that if $X \subseteq E(G)$ is an $sA$-cut in $G$, then the
  set $X'=\SB z_e \SM e \in X\SE$ is a $\cF$-transversal in $G'$. This
  is because any path from $K(s)$ to $\bigcup_{a \in A}K(a)$ in
  $G'-X'$ corresponds to a path from $s$ to $A$ in $G$. Similarly, if
  $X'\subseteq V(G')$ is an $\cF$-transversal in $G'$ of size at most
  $2q$, then the set $X=\SB e \SM z_e \in X'\SE$ is an $sA$-cut in
  $G$. To see this, first observe that $|X'|\leq 2q$ implies
  that $X'$ cannot contain all $2q+1$ vertices in $K(a)$ for any
  vertex $a \in V(G)$ and 
  therefore any path from $s$ to $A$ in $G-X$ corresponds to a path
  from $K(s)$ to $\bigcup_{a \in A}K(a)$ in $G'-X'$.

  Let $W' \subseteq V(G')\setminus T$ be the set of vertices of $G'$
  satisfying the properties stated in Theorem~\ref{thm:shadow-cover},
  i.e.\ if there is a $\cF$-transversal of size at most $2q$ in $V(G') \setminus K(s)$,
  then there is an $\cF$-transversal $R$ in $V(G') \setminus (W' \cup K(s))$ 
  of size at most $2q$ such that every vertex $v \notin W \cup R$
  is connected to $K(s)$ in $G' - R$. We claim that the set $W=\SB a \in
  V(G) \SM K(a)\subseteq W'\SE$ satisfies the properties given in the
  statement of the lemma.
  First note that without loss of generality we may assume that $z_e \in W'$ for every
  edge $e=uv \in E(G)$ such that $u,v \in W$; this is because $z_e$
  can never be reachable from $K(s)$ if $K(u)$ and $K(v)$ are both contained
  in $W'$.
  
  Because $W' \subseteq V(G')\setminus K(s)$, we obtain that $W
  \subseteq V(G)\setminus \{s\}$. Moreover, because $Y$ is an
  $sA$-cut of size at most $2|Z\setminus \comp{Z}|\leq 2q$ in $G$ (using~\Cref{lem:deleted-edges}), we obtain that the set $R=\SB z_e
  \SM e \in Y\SE$ is a $\cF$-transversal of size at most $2q$ in
  $V(G')\setminus K(s)$. Therefore, by Theorem~\ref{thm:shadow-cover}
  there is a $\cF$-transversal $R' \subseteq V(G') \setminus (W' \cup K(s))$ 
  of size at most $2q$ such that every vertex $v \notin W' \cup R'$
  is connected to $K(s)$ in $G' - R'$. Note that without loss of generality we may
  assume that $R'$ is a minimal $\cF$-traversal; otherwise we can
  remove unnecessary vertices from $R'$ to make it minimal without
  changing the properties of $R'$. Therefore, 
  $R' \subseteq \SB z_e \SM e \in E(G)\SE$ also holds.
  Let $Y'$ be the set $\SB e \in E(G) \SM z_e \in
  R'\SE$. Because $R'$ is a $\cF$-transversal of size at most $2q$ in
  $G'$,
  it holds that $Y'$ is an $sA$-cut of size at most $2q$ in
  $G$. Moreover, because $R' \subseteq V(G') \setminus (W' \cup K(s))$
  and, as we observed above, $W'$ contains $z_e$ for every edge
  $e=uv$ with $u,v\in W$,
  we obtain that
  $Y' \subseteq E(G)\setminus E(G[W])$. Since every vertex $v \notin
  W'\cup R'$ is connected to $K(s)$ in
  $G'-R'$, we see that every vertex $u \notin W$ is connected to $s$
  in $G-Y'$ showing \textbf{1.} Moreover, because $Y'$ is a minimal $sA$-cut we
  also obtain $Y'\subseteq \delta(W)$, which shows \textbf{2.}
  Since \textbf{3.} follows directly from Lemma~\ref{lem:anti-sol}, it only
  remains to show \textbf{4.}--\textbf{7.}.

  Let $C$ be a component of $G[W]$. Because $Y'\cap E(C)=\emptyset$
  and every vertex in $N(C)$, where $N(C)$ denotes the set of all
  neighbours of $C$ outside of $C$,  is reachable from $s$ in
  $G-Y'$, it follows that if $\delta(C) \setminus Y'\neq
  \emptyset$, then all vertices in $C$ are reachable from $s$ in
  $G-Y'$. Therefore, because $Y'$ is an inclusion-wise minimal
  $sA$-cut in $G$, we can assume that either
  $Y'\cap\delta(C)=\emptyset$ or $\delta(C)\subseteq Y'$;
  otherwise removing the edges in $\delta(C)$ from $Y'$ gives again
  an $sA$-cut showing that $Y'$ is not inclusion-wise minimal. This shows \textbf{4.}
  Therefore, if $t \in C$ then $\delta(C)\subseteq Y'$
  since otherwise $Y'$ is would not be an $sA$-cut in
  $G$, and this shows \textbf{5.} Item \textbf{6.} can now be shown similarly: if a component $C$ contains two distinct vertices $x_C$
  and $x_{C'}$ for some variable $x$ of $S$, then $\delta(C)\subseteq Y'$
  because otherwise $Y'$ is not conformal. Moreover, the same
  applies if $C$ contains some $x_C$ for some decided
  variable $x$ w.r.t. $Y'$, which shows \textbf{7.}.
\end{proof}

\newcommand{\trip}{\mathcal{T}}
\begin{algorithm}[htb]
  \caption{Method for branching.} \label{alg:branch}
  \small
  \begin{algorithmic}[1]
    \INPUT a simple $\minlin{2}{\ZZ_{p^n}}$-instance $S$,
    integers $k \geq q \geq 0$
    \OUTPUT the set $\cY$ of conformal cuts
    \Function{\textsc{branch}}{$S$, $k$, $q$}
    \State $G \gets G(S)$
    \State Call Lemma~\ref{lem:our-shadow-cover} on $(G,q)$,
    producing $W \subseteq V(G) \setminus \{s\}$

    \If {$\delta(W)$ is not conformal}
    \State \Return $\emptyset$
    \EndIf
    \State $\CCC_0 \gets$ connected components of $G[W]$
    \State Let $\CCC^- \subseteq \CCC_0$ contain every component $C$ \label{algline:ccc-}
    such that
    $t \in C$, or
    there exist $x \in V(S)$ and distinct $\alpha, \beta \in
      \Gamma^{\neq 0}$ such that $x_\alpha, x_\beta \in (V(G)
      \setminus W) \cup C$
    \State $V_0 \gets \bigcup_{C \in \CCC^-} C$
    \State \Return \Call{branchUndecided}{$\CCC_0 \setminus \CCC^-$, $k-q$, $2q-|\delta(V_0)|$,
      $V(G) \setminus W$, $V_0$}
    \EndFunction
  \end{algorithmic}
\end{algorithm}

\begin{algorithm}[htb]
  \caption{Method for branching on undecided variables.} \label{alg:branchundecided}
  \small
  \begin{algorithmic}[1]
    \INPUT the set $\CCC$ of remaining components, the
    remaining budgets $k'$ and $b$, disjoint vertex sets $V_1, V_0 \subseteq V(G)$
    \Function{\textsc{branchUndecided}}{$\CCC$, $k'$, $b$, $V_1$, $V_0$}
    \If {$b < 0$}
    \State \Return $\emptyset$
    \EndIf
    \If {there exist $C \in \CCC$, $x \in V(S)$ and distinct $\alpha, \beta \in \Gamma^{\neq 0}$
      such that $x_\alpha, x_\beta \in V_1 \cup C$} \label{algline:undecided-single}
    \State \Return \Call{branchUndecided}{$\CCC \setminus \{C\}$, $k'$, $b-|\delta(C)|$, $V_1$, $V_0 \cup C$}
    \EndIf
    \If {there exist $C_1, C_2 \in \CCC$, $x \in V(S)$, $\alpha, \beta \in \Gamma^{\neq 0}$ with $\alpha \neq \beta$
    and $x_\alpha, x_\beta \in C_1 \cup C_2$} \label{algline:undecided-pair}
    \State $\cY \gets $ \Call{branchUndecided}{$\CCC \setminus \{C_1\}$, $k'$, $b-|\delta(C_1)|$, $V_1$, $V_0 \cup C_1$}
    \State $\cY \gets \cY \; \cup  $ \Call{branchUndecided}{$\CCC \setminus \{C_1,C_2\}$, $k'$, $b-|\delta(C_2)|$, $V_1 \cup C_1$, $V_0 \cup C_2$}
    \State \Return $\cY$
    \EndIf
    \State $\CCC_U \gets \,$ \Call{getUnSatisfied}{$\CCC$} \label{algline:undecided-post}
    \State \Return \Call{branchUnSatisfied}{$\CCC_U$, $k'$, $b$, $V_0$}
%
    \EndFunction
  \end{algorithmic}
\end{algorithm}

\begin{algorithm}[htb]
  \caption{Method for branching on unsatisfied components.} \label{alg:branchunsat}
  \small
  \begin{algorithmic}[1]
    \INPUT the set $\CCC_U$ of unsatisfied components, the
    remaining budgets $k'$ and $b$, a set $V_0 \subseteq V(G)$
    \Function{\textsc{branchUnSatisfied}}{$\CCC_U$, $k'$, $b$, $V_0$}
    \If {$k'<0$ or $b<0$}
    \State \Return $\emptyset$
    \EndIf
    \If {$\CCC_U = \emptyset$}
    \State \Return $\{\delta(V_0)\}$
    \EndIf
    \State Let $C \in \CCC_U$
    \State $\cY \gets $ \Call{branchUnSatisfied}{$\CCC_U \setminus \{C\}$, $k'-1$, $b$, $V_0$}
    \State $\cY \gets \cY \; \cup \; $ \Call{branchUnSatisfied}{$\CCC_U \setminus \{C\}$, $k'$, $b-\delta(C)$, $V_0 \cup C$}
    \State \Return $\cY$
%
    \EndFunction
  \end{algorithmic}
\end{algorithm}

The following lemma now uses the set $W \subseteq V(G)$ computed
in~\Cref{lem:our-shadow-cover} to compute a set of at most
$2^{\bigoh(k)}$
conformal cuts $\cY$ which contains a cut $Y$
with the following property.
If $S$ has a solution $Z$ of size at most $k$ such that $|Z\setminus \comp{Z}|=q$, 
then there is a subset $Z_2 \subseteq (Z \setminus \comp{Z})$ with 
for which $Y$ serves as a $2$-approximation:
more precisely, $|Y| \leq 2|Z_2|$,
$(Z \setminus Z_2) \cup Y$ is a solution to $S$, and
$S - ((Z \setminus Z_2) \cup Y)$ admits a satisfying assignment
that agrees with $Y$.
Observe that the cost of $S - Y$ is at most $|Z \setminus Z_2| \leq k - |Y|/2$.
Now~\Cref{lem:get-class-assignment} is an immediate consequence
of~\Cref{lem:list-conformal-cuts}, i.e.\ instead of returning
the set $\cY$ of conformal cuts, we choose one conformal cut $Y \in
\cY$ uniformly at random and output the class assignment corresponding
to $Y$.
The main idea
behind computing $\cY$ is to use the fact that we only need to
consider conformal cuts that are within $\delta(W)$ and this
allows us to branch on which components of $G[W]$ are
reachable from $s$ (see also Property \textbf{4.} in~\Cref{lem:our-shadow-cover}).

\begin{lemma}
  \label{lem:list-conformal-cuts}
  Let $S$ be a simple instance of
  $\lin{2}{\RZ}$ and let $k$ and $q$ with $k\geq q$ be integers.
  There is a randomized algorithm that 
  takes $(S,k,q)$ 
  as input and returns 
  in
  $\bigohs(2^{\bigoh(k)})$ time
  a set $\cY$ of at most
  $2^{\bigoh(k)}$ conformal cuts (each of size at most $2q$),
  such that if $S$ has a solution $Z$ of size at most $k$ such that $q=|Z\setminus \comp{Z}|$
  then with probability at least $2^{-O(k+q^2)}$
  there is a cut $Y \in \cY$ with the following property:
  there is a partition $Z=Z_1 \cup Z_2$ where $\comp{Z} \subseteq Z_1$ and $|Y| \leq 2|Z_2|$,
  and an assignment that satisfies $S-(Z_1 \cup \con(Y))$ and agrees with $Y$.
\end{lemma}
\begin{proof}
  The algorithm is shown as Algorithm~\ref{alg:branch}--\ref{alg:branchunsat}.
  It returns the set $\cY$ of conformal cuts and
  is initially called with \textsc{branch}($S$, $k$, $q$)
  where $k$ and $q$ are as in the statement of this lemma.
  It starts by computing the set $W\subseteq V(G)\setminus
  \{s\}$ from $(G,q)$ using~\Cref{lem:our-shadow-cover} such
  that with probability $2^{-O(q^2)}$ there is a (minimal)
  $sA$-cut $Y'$ of size at most $2q$ satisfying \textbf{1.}--\textbf{7.}.
  Given
  $W$ and the fact that $Y'$ has to satisfy \textbf{1.}--\textbf{7.},
  the algorithm now branches over all possibilities for $Y'$
  by
  guessing which components of $G[W]$ will be reachable from $s$ in
  $G-Y'$. In particular, using \textbf{2.} and \textbf{4.}, it follows that
  $Y'=\delta(\CCC)$, where $\CCC$ is the set of all components of $G[W]$
  that are not reachable from $s$ in $G-Y'$. Moreover, because of \textbf{3.},
  $Y'$ is a conformal $sA$-cut in $G$. If $Y' \subseteq \delta(W)$
  then $\delta(W)$ is also conformal.
  Therefore, if $\delta(W)$ is not conformal, the algorithm can correctly reject the choice of $W$ (and return $\cY=\emptyset$). In the following we assume that $\delta(W)$ is conformal,
  i.e.\ $t \in W$ and for every variable $x$ of $S$ at most one vertex
  $x_\alpha$ is not in $W$.
  We say that $x$ is \emph{$W$-decided} if there is such a class and
  otherwise we say that $x$ is \emph{$W$-undecided}.
  Let $\CCC_0$ be the set
  of all components of $G[W]$ and let $\CCC^-$ be the subset of $\CCC_0$
  consisting of all components $C$ such that either:
  \begin{itemize}
  \item $t \in C$,
  \item $C$ contains $x_\alpha$ and $x_{\alpha'}$ for some variable
    $x$ of $S$ and two distinct classes $\alpha$ and $\alpha'$, or
  \item $C$ contains $x_\alpha$ for some $W$-decided variable $x$ of $S$.
  \end{itemize}
  Because $Y'$ has to be a conformal $sA$-cut, no
  component in $\CCC^-$ can be reachable from $s$ in $G-Y'$ so $\delta(\CCC^-)\subseteq Y'$, where $\delta(\CCC^-)=\bigcup_{C \in
    \CCC^-}\delta(C)$. Consequently, if $|\delta(\CCC^-)|>2q$, then
  the algorithm can correctly reject the choice of $W$ by returning
  $\cY=\emptyset$. Otherwise,
  every component $C$ in $\CCC_1=\CCC_0\setminus \CCC^-$ satisfies:
  \begin{itemize}
  \item $t \notin C$,
  \item $C$ does not contain $x_\alpha$ for any $W$-decided variable $x$
    of $S$,
  \item $C$ contains at most one vertex $x_\alpha$ for every $W$-undecided
    variable $x$ of $S$.
  \end{itemize}
  Next, the algorithm branches on which components in
  $\CCC_1$ are separated from $s$ in a solution. This
  part of the algorithm, which is given in Algorithm~\ref{alg:branchundecided},
  returns the set $\cY$ of conformal cuts and
  is called with \textsc{branchUndecided}($\CCC_1$, $k'=k-q$, $b=2q-|\delta(\CCC^-)|$, $V_1=V(G) \setminus W$, $V_0=\bigcup \CCC^-$).
  It is a recursive algorithm, where the first argument is a set of pending components of $G[W]$
  (where the choice whether to separate them from $s$ has not been taken yet),
  the second and third are the remaining budgets from the split $k=k'+q$ of $k$ that we were called with,
  and the last two are sets $V_1, V_0 \subseteq V(G)$ of vertices that we have committed to not separating, respectively separating from $s$ in the current branch. 
  The algorithm branches until $\delta(V_0)$ is a conformal cut, at which point
  it hands over to \textsc{branchUnSatisfied}, given in Algorithm~\ref{alg:branchunsat},
  for the last step of the branching (which is where the argument $k'$ is required).
  The main branching step is in line~\ref{algline:undecided-pair}: it selects two components
  $C_1$, $C_2$ from the pending set of components which conflict with each other,
  in the sense that there is a variable $x$ and distinct classes $\alpha, \beta \in \Gamma^{\neq 0}$
  such that $x_\alpha \in C_1$ and $x_\beta \in C_2$, and branches on
  (1) separating $C_1$ from $s$, or (2) not separating $C_1$, in which case
  it must separate $C_2$ from $s$. In the second branch, it can happen
  that there are further components $C' \in \CCC$ that conflict with $C_1$, 
  which are then forced to be deleted in the recursive call
  by the check on line~\ref{algline:undecided-single}.
  We make two quick observations about the algorithm's behaviour. 
  
  \begin{claim} \label{claim:alg-gives-conformal}
    If a call of \textsc{branchUndecided} reaches line~\ref{algline:undecided-post},
    with arguments $(\CCC, k', b, V_1, V_0$),
    then $\delta(V_0)$ is a conformal cut of cost $|\delta(V_0)|=2q-b$.
  \end{claim}
  \begin{claimproof}
    We claim the following invariants of all calls to \textsc{branchUndecided}:
    $t \in V_0$, $s \in V_1$, and $\delta(V_1)$ is a conformal cut. 
    This is true in the initial call from \textsc{branch},
    since initially $V_1=V(G) \setminus W$ and these properties has been verified
    in \textsc{branch}. Furthermore, we only place a component $C$ in
    $V_1$ under the check on line~\ref{algline:undecided-pair},
    and at this point it has been verified that $C$ does not conflict with $V_1$
    in line~\ref{algline:undecided-single}. Thus the invariant holds.
    Now assume for a contradiction that we reach line~\ref{algline:undecided-post}
    but $\delta(V_1)$ is not a conformal cut. Since $t \in V_1$
    and $s \in V_0$, this must imply that there are two vertices
    $x_\alpha$, $x_\beta$ for some variable $x \in V(S)$ 
    which are not in $V_0$. At most one of them is in $V_1$ by the invariant.
    But then either (w.l.o.g.) $x_\alpha \in V_1$ and $x_\beta \in C$
    for some $C \in \CCC$, in which case line~\ref{algline:undecided-single} applies,
    of $x_\alpha \in C_1$, $x_\beta \in C_2$ for some $C_1, C_2 \in \CCC$,
    in which case line~\ref{algline:undecided-pair} applies.
    Thus $\delta(V_0)$ is a conformal cut.
    The final statement follows since $\delta(V_0)$ is the disjoint union of $\delta(C)$
    for all components $C$ placed in $V_0$, and the way the algorithm traces the $b$-variable. 
  \end{claimproof}

  \begin{claim} \label{claim:alg-finds-V0}
    Let $\CCC' \subseteq \CCC$ be a set of components such that $Y'=\delta(\CCC')$
    is a conformal cut with $|Y'|\leq 2q$. Then there is a call made to \textsc{branchUndecided}
    which reaches line~\ref{algline:undecided-post}
    with values $V_0, V_1$ such that $\delta(V_0) \subseteq Y'$ and $\delta(V_1) \cap Y' = \emptyset$.
  \end{claim}
  \begin{claimproof}
    Let $V_Y=\bigcup \CCC'$ be the set of vertices separated from $s$ by $Y'$.
    We claim, stepping recursively through all calls to \textsc{branchUndecided},
    that we can maintain the invariants $V_0 \subseteq V_Y$, $V_1 \cap V_Y=\emptyset$
    and $b \geq 2q-|Y'|$ in at least one branch of the calling tree. For the initial call,
    the components $\CCC^-$ from line~\ref{algline:ccc-} of \textsc{branch} must satisfy $\CCC^- \subseteq \CCC'$
    since $Y'$ is conformal, so $V_0 \subseteq V_Y$ holds in the root call to \textsc{branchUndecided}.
    Furthermore $V_1=V(G) \setminus W$ in this call is disjoint from $V_Y$.
    Consider then some instance of \textsc{branchUndecided} in which these invariants hold,
    and assume first that the check on line~\ref{algline:undecided-single} applies
    for some $C \in \CCC$. Then $C$ conflicts with $V_1$, and since $V_1 \cap V_Y = \emptyset$ by assumption,
    we must have $C \subseteq V_Y$. Hence the invariant is maintained.
    Next, assume that line~\ref{algline:undecided-pair} applies for some pair $C_1, C_2 \in \CCC$.
    If $C_1 \subseteq V_Y$, then the invariant is maintained in the first recursive call made.
    Otherwise $C_1 \cap V_Y = \emptyset$, and $C_2 \subseteq V_Y$ since $Y'$ is conformal,
    hence the invariant is maintained in the second call. 
    The invariant for the value of $b$ follows from the value of $V_0$. 
    Hence, there is a call where line~\ref{algline:undecided-post} of \textsc{branchUndecided}
    is reached with the invariant intact. 
  \end{claimproof}

  Next, we consider the work of the algorithm after the point where $\delta(V_0)$ is a conformal cut, 
  where it decides whether to make further cuts in the instance or not.
  For this step, say that a component $C \in \CCC$ is \emph{self-satisfiable} if the
  equations of $S$ corresponding to edges of $G[C]$ 
  can be satisfied using an assignment that assigns every variable $x$
  with vertex $x_\alpha$ in $C$ to a value in $\alpha$; note that
  because $C \in \CCC_1$ every variable $x$ with a vertex $x_\alpha$
  in $C$ has exactly one such vertex $x_\alpha$ in $C$ and therefore
  $\alpha$ is uniquely defined. Let \textsc{GetUnSatisfied$(\CCC)$} on
  line~\ref{algline:undecided-post} return the set $\CCC_U$ of all components $C$ from $\CCC$
  that are not self-satisfiable. To check whether $C \in \CCC$ is self-satisfiable,
  we check $\lin{2}{\ZZ_{p^n}}$ on the instance $S_C$ consisting of all equations
  from $S$ between variables represented in $C$, together with,
  for every vertex $x_\alpha$ in $C$, a constraint forcing $x \in \alpha$.
  This can be done as follows. Let $\alpha$ contain those elements from $\ZZ_{p^n}$
  which can be written as $ap^b + cp^{b+1}$ for $a \in [p-1]$, $b \in [n]$ and $c \in \ZZ$.
  Every class $\alpha$ can be written this way. 
  Then $x \in \alpha$ corresponds to the unary constraint $p^{n-b-1}x=ap^{n-1}$.
  Satisfiability of the resulting system can be checked in polynomial time;
  see, e.g.~\cite[p. 473]{Arvind:Vijayaraghavan:stacs2005}). 
  Hence \textsc{GetUnSatisfied} runs in polynomial time.
  Finally, \textsc{branchUnSatisfied} simply branches over all components of $\CCC_U$
  as being separated from $s$ or not, while tracking the budgets $k'$ and $b$.
  (In particular, if $C \in \CCC_U$ and $\delta(C) \not \subseteq Y$, 
  then every solution corresponding to a conformal cut $Y \subseteq \delta(W)$
  must delete an equation represented in $C$.)
  
  This completes the description. Let us consider correctness. The output $\cY$ of \textsc{branch}
  contains only cuts $Y$ with $|Y| \leq 2q$ by design. Furthermore, any such cut is conformal
  by Claim~\ref{claim:alg-gives-conformal}, since cuts $Y$ are produced only in \textsc{branchUnSatisfied}
  after line~\ref{algline:undecided-post} of \textsc{branchUndecided} has been passed.
  Let $Z$ be a solution to $S$ of cost at most $k$ and with $|Z\setminus \comp{Z}|=q$.
  We need to show that assuming $W$ was constructed correctly,
  the output of \textsc{branch}$(S, k, q)$ contains a cut $Y$
  such that $S-(\comp{Z} \cup \con(Y))$ has a satisfying assignment that agrees with $Y$.
  Let $Y' \subseteq \delta(W)$ be the cut guaranteed by Lemma~\ref{lem:our-shadow-cover}.
  Let $V_Y$ be the vertices separated from $s$ by $Y'$.
  By Claim~\ref{claim:alg-finds-V0}, there is a branch in which line~\ref{algline:undecided-post}
  of \textsc{branchUndecided} is reached with $V_0 \subseteq V_Y$,
  and by Claim~\ref{claim:alg-gives-conformal} any output $Y$ from this point onwards is a conformal cut. 
  Furthermore, for every $C \in \CCC_U$ 
  either $Y'$ separates $C$ from $s$ or else $\comp{Z}$ contains an equation from within $C$. 
  It follows that we can step through \textsc{branchUnSatisfied} such that 
  we get a cut $Y=\delta(V_0)$ such that $V_0 \cap (\bigcup \CCC_U)=V_Y \cap (\bigcup \CCC_U)$
  while keeping $k', b \geq 0$.
  Let $Y$ be such a cut. We have $Y \subseteq Y'$, and if $Y \neq Y'$, then the difference
  is accounted for by additional self-satisfiable components $C$ separated from $s$ by $Y'$. 
  We define an assignment $\varphi'$
  as follows. We begin with $\varphi'(x)=\varphi(x)$ for any vertex $x_\alpha$ reachable
  from $s$ in $G-Y'$, and $\varphi'(x)=0$ for any variable $x$ with no vertex reachable from $s$ in $G-Y$.
  For the rest, let $C$ be a self-satisfiable component reachable from $s$ in $G-Y$
  and let $\varphi_C$ be a satisfying assignment to $C$ that agrees with $Y$. 
  Let $C_Z \subseteq C$ be the vertices in $C$ whose variables are zero in $\varphi$. 
  We let $\varphi'$ match $\varphi$ on $C \setminus C_Z$ and match $\varphi_C$ on $C_Z$. 
  Then $\varphi'$ is an assignment that agrees with $Y$. Let $Z'$ be the equations
  violated by $\varphi'$. Then $Z' \setminus Z \subseteq \con(Y)$: the only ways we have modified
  $\varphi$ is to set some variables to $0$ (due to Lemma~\ref{lem:our-shadow-cover}
  pushing $Z$ to cut closer to $s$), which is handled by $Y$,
  and to change $C_Z$ away from $0$, but this only affects edges in $\delta(C_Z)$
  which are already violated by $\varphi$. 
  Let $Z_1 \subseteq Z$ be all equations in $Z$ not handled by $Y$, i.e.\ 
  all equations with at least one edge for which both endpoints are reachable from $s$ in $G-Y$.
  Then $Z_1 \cup \con(Y)$ is a solution, and $|Y| \leq 2|Z \setminus Z_1| \leq 2q$.

  The running time consists of $O^*(4^{2q})$ time for the call to Lemma~\ref{lem:our-shadow-cover},
  followed by a branching procedure which takes polynomial time for every branching step,
  branches into at most two directions each time, and has depth at most $k'+2q \leq 2k$. 
  Hence the time is bounded as $O^*(2^{O(k)})$. The claim $|cY| = 2^{O(k)}$ follows from this.
  Finally, the only randomness of the algorithm is the call to Lemma~\ref{lem:our-shadow-cover},
  hence success probability is $2^{O(q^2)}$ by this lemma.
\end{proof}

\section{Hardness of \FPT-Approximation}
\label{sec:hardness}

We complement the approximation algorithm
with hardness results:
we prove (1) that for finite, commutative, non-trivial
rings $R$,
$\minlin{r}{R}$ is \W{1}-hard to FPT-approximate within any constant when $r \geq 3$
and (2) the existence of finite commutative rings $R$ such that
$\minlin{2}{R}$ is \W{1}-hard to FPT-approximate within any constant.

Let $G$ denote an arbitrary Abelian group.
An expression $x_1+\dots+x_r=c$ is an {\em equation over} $G$ if
$c \in G$ and $x_1,\dots,x_r$ are either variables or inverted variables with domain $G$.
We say that it is an {\em $r$-variable equation} if it contains at most $r$
distinct variables. We sometimes consider the natural group-based variants of
the $\lin{3}{R}$ and $\minlin{r}{R}$ problems in what follows.
We begin by proving the following result:
$\minlin{3}{G}$ is \W{1}-hard to FPT-approximate within any constant when $G$ is a
non-trivial Abelian group. This directly implies that $\minlin{3}{R}$ is 
\W{1}-hard to FPT-approximate within any constant when $R$ is a non-trivial ring
since $(R, +)$ is an Abelian group.
This result strengthens two previously known hardness results:
(I) $\minlin{3}{G}$ is \W{1}-hard and (II) $\minlin{3}{G}$ is \NP-hard to
approximate within any constant.
Result (I) follows immediately from
Theorem 6.1 in \cite{Dabrowski:etal:soda2023}.
Result (II) can be derived as follows.
Håstad~\cite{haastad2001some} proves that for every $\eps > 0$, it is \NP-hard to distinguish instances of
$\lin{3}{G}$
(where $G$ is a non-trivial Abelian group), that are $(1/|G| + \eps)$-satisfiable from those that
are $(1-\eps)$-satisfiable.
Assume that $\minlin{3}{G}$ is $c$-approximable in polynomial time for some constant $c$. 
Pick $\eps > 0$ such that $(1-c \cdot \eps) \geq (1/|G| + \eps)$. Take an instance $S$ of $\lin{3}{G}$, set $k = \eps \cdot |S|$, run the factor-$c$ approximation algorithm on $S$, and deduce
the correct answer from the output of the algorithm. 
The starting point for our reduction is the following problem
and the corresponding hardness result.
For a vector $x \in \FF^{n}$ in a field $\FF$,
the \emph{Hamming norm} $||x||_0$ is 
the number of nonzero components in $x$.
In the \textsc{Maximum Likelihood Decoding Problem over $\FF_p$} (\MLD{}),
we are given a matrix $A \in \FF^{n \times m}_p$ and a vector $b \in \FF^m_p$,
and the goal is to
find $x \in \FF_p^m$ of minimum Hamming weight such that $Ax = b$.

\begin{theorem}[Theorem~5.1~in~\cite{bhattacharyya2021parameterized}]
  \label{thm:gapmld-is-hard}
  For every prime $p$, \MLD{} parameterized by 
  the minimum Hamming weight of the solution 
  is \W{1}-hard to approximate within any constant factor.
\end{theorem} 


Before we start the reduction, let us recall some facts about Abelian groups.
A \emph{cyclic} group is generated by a single element and 
every finite cyclic group $C_n$ of order $n$ is isomorphic to the additive group of ${\mathbb Z}_n$.
The \emph{direct sum} of two groups
$\GG_1 = (G_1, +_1)$ and $\GG_2 = (G_2, +_2)$ 
is a group $\GG_1 \oplus \GG_2$ with pairs 
$\{ (g_1, g_2) \mid g_1 \in G_1, g_2 \in G_2 \}$ as elements,
and with group operation applied componentwise,
i.e.\ $(g_1, g_2) + (g'_1, g'_2) = (g_1 +_1 g'_1, g_2 +_2 g'_2)$.
Cyclic groups are the building blocks of
more complex Abelian groups: the fundamental theorem of finite Abelian groups asserts that
every finite Abelian group is a direct sum of cyclic groups whose orders are prime powers.
Hence,
we start by proving hardness of $\minlin{3}{C_p}$, where $p$ is prime.
For further convenience, we define a restricted class of instances 
and prove hardness for this case.

\begin{definition}
  An instance of $\minlin{3}{G}$
  is \emph{almost homogeneous} if every equation in it is of the form
  \begin{itemize}
    \item $x + y + z = 0$ and crisp, or
    \item $x = a$ for some $a \in G$ and crisp, or
    \item $x = 0$.
  \end{itemize}
\end{definition}

\begin{lemma} \label{lem:min3lin-special}
  For every prime $p$,
  $\minlin{3}{C_p}$
  is \W{1}-hard to approximate within any constant factor 
  even when restricted to almost homogeneous instances.
\end{lemma}
\begin{proof}
  Let $(A, b)$ be an instance of \MLD{} and let $k$ be the parameter.
  As usual, we let $a_{ij}$ denote the entry in row $i$ and column $j$ of $A$.
  Construct an almost homogeneous instance of $(S, k)$ of 
  $\minlin{3}{C_p}$ as follows.
  Start by introducing primary variables $x_1, \dots, x_n$ to $V(S)$,
  and adding soft constraints $x_i = 0$ for all $i \in [n]$.
  For every $a \in \range{0}{p-1}$,
  introduce auxiliary variables $z_a$ and
  add crisp equations $z_a = a$ to $S$.
  Consider a row equation $\sum_{j=1}^{n} a_{ij} x_{j} = b_{i}$.
  Since we are working over the field $\ZZ_p$,
  this equation can be written as
  \begin{equation} \label{eq:row-equation-additive}
    \underbrace{x_{1} + \dots + x_{1}}_{a_{i1} \textrm{ times}} 
    + \dots +
    \underbrace{x_{n} + \dots + x_{n}}_{a_{in} \textrm{ times}}
     - b_{i} = 0 \mod p.
  \end{equation}
  Let $L = \sum_{j=1}^{n} a_{ij}$ and
  define $s_1, \dots, s_L$,
  where $s_1 = \dots = s_{a_{i1}} = x_1$,
  $s_{a_{i1}+1} = \dots = s_{a_{i2}} = x_2$,
  and so on.
  Note that $\sum_{\ell=1}^{L} s_\ell = \sum_{j=1}^{n} a_{ij} x_j$.
  To express~\eqref{eq:row-equation-additive} 
  as a system of homogeneous $3$-variable equations,
  we introduce auxiliary variables
  $u_1, \bar{u}_1 \dots, u_L, \bar{u}_L$, and
  add crisp equations
  $u_\ell + \bar{u}_{\ell} + z_{0} = 0$ for all $\ell$.
  Note that these constraints force $u_\ell$ and $\bar{u}_{\ell}$ 
  to be additive inverses modulo $p$.
  Add crisp equations
  \begin{equation} \label{eq:chop}
    \begin{aligned}
      z_0 + s_1 + u_1 &= 0 \mod p, \\
      \bar{u}_{\ell-1} + s_{\ell} + u_{\ell} &= 0 \mod p \quad \text{ for } \ell \in \range{2}{L-1}, \\
      \bar{u}_{L-1} + s_{L} + z_{-b_i} &= 0 \mod p.
    \end{aligned}
  \end{equation}
  Observe that the sum of the equations above, after cancellation of auxiliary variables,
  implies $\sum_{\ell=1}^{L} s_\ell + z_{-b_i} = \sum_{i=1}^{n} a_{ij} x_i - b_i = 0$.
  Applying the same reduction to all rows of $A$,
  we construct $(S, k)$ in polynomial time.

  For correctness, first assume that $(A, b)$ is a yes-instance,
  i.e.\ there exist $x_1, \dots, x_n \in \ZZ_p$
  that satisfy row equations $\sum_{j=1}^{n} a_{ij} x_{j} = b_{i}$ for all $i \in [n]$
  and at most $k$ values in $(x_1, \dots, x_n)$ are nonzero.
  Observe that it can be extended to auxiliary variables $u_i$, $\bar{u}_i$
  in a unique way.
  The obtained assignment satisfies all crisp and all but $k$ soft equations $x_i = 0$.

  For the other direction, suppose $(S, k)$
  admits an assignment $\alpha$ that 
  violates at most $c \cdot k$ constraints
  for some constant $c$.
  Define vector $\alpha(x) = (\alpha(x_1), \dots, \alpha(x_n))$.
  Since all soft constraints in $S$ are of the form $x_i = 0$, 
  we have $||\alpha(x)||_0 \leq c \cdot k$.
  By construction of $S$, we have $A \alpha(x) = b$.
  Hence, if $\minlin{3}{C_p}$ admits an $O^*(f(k))$-time $c$-approximation,
  then \MLD{} also admits an $O^*(f(k))$-time $c$-approximation as well.
\end{proof}

\begin{theorem} \label{thm:gap-min3lin(group)-is-hard}
  For every nontrivial Abelian group $G$,
  $\minlin{3}{G}$ is \W{1}-hard
  to approximate within any constant factor.
\end{theorem}
\begin{proof}
Assume without loss of generality that
$G \cong \bigoplus_{i=1}^{t} C_{p_i^{\ell_i}}$ and set
$p = p_1$ and $\ell = \ell_1$.
We present an \FPT-reduction from $\minlin{3}{C_p}$
to $\minlin{3}{G}$. 
Due to Lemma~\ref{lem:min3lin-special}, 
it is sufficient to consider an almost homogeneous 
instance $(S, k)$ of $\minlin{3}{C_p}$.
Create an instance $(S', k)$ of $\minlin{3}{G}$ as follows.
Let $V(S')$ contain variable $v', v''$ for every $v \in V(S)$.
Add crisp equations equivalent to
\begin{equation} \label{eq:multiplicity}
  \underbrace{v' + \dots + v'}_{p^{\ell-1}} = v''    
\end{equation}
using the same reduction from many summands to three as
in Equation~\eqref{eq:chop} in the proof of Lemma~\ref{lem:min3lin-special}.
For every equation of the form $x + y + z = 0$ in $S$,
add $x'' + y'' + z'' = 0$ to $S'$.
For every equation of the form $x = a$ in $S$,
add $x'' = (p^{\ell-1} a, 0, \dots, 0)$ in $S'$.
This completes the construction.

For one direction,
consider an arbitrary assignment $\alpha : V(S) \to C_p$
and define $\alpha' : V(S') \to G$
by setting $\alpha(v')$ equal to the element
$(\alpha(v), 0, \dots, 0)$
and $\alpha(v'')$ to the element 
$(p^{\ell-1} \alpha(v), 0, \dots, 0)$.
Note that $\alpha'$ satisfies equations~\ref{eq:multiplicity}
for all $v', v'' \in V(S')$.
Furthermore, if
$\alpha(x) + \alpha(y) + \alpha(z) = 0 \bmod p$, then
$p^{\ell-1} (\alpha(x) + \alpha(y) + \alpha(z)) = 0 \bmod p^\ell$
and
$\alpha'(x'') + \alpha'(y'') + \alpha'(z'') = 
(0, \dots, 0) = 0$.
Finally, if $\alpha(x) = a$, then
$\alpha'(x'') = (p^{\ell-1} a, 0, \dots, 0)$.
We conclude that $\alpha'$ violates
at most as many equations as $\alpha$.

For the other direction,
consider an arbitrary assignment $\beta' : V(S') \to G$
and define $\beta : V(S) \to C_{p^l}$
by setting $\beta(v)$ to the projection
of $\beta'(v')$ onto the first component.
By the crisp equations~\eqref{eq:multiplicity}, 
the first component of $\beta(v'')$ equals $p^{\ell-1} \beta(v)$.
Thus, if $\beta(x'') = (p^{\ell-1} a, 0, \dots, 0)$, then $\beta(x) = a$.
Moreover, if $\beta(x'') + \beta(y'') + \beta(z'') = 0$,
then the equation holds in the first component,
hence $p^{\ell-1} (\beta(x) + \beta(y) + \beta(z)) = 0 \bmod p^\ell$
and $\beta(x) + \beta(y) + \beta(z) \bmod p$.
Thus, $\beta'$ violates at most as many equations as $\beta$.
\end{proof}

We obtain the following
due to the additive group of every ring being Abelian.

\begin{theorem} 
 \label{thm:min3lin-hard}
 Let $R$ be a non-trivial finite ring.
 $\minlin{r}{R}$ is \W{1}-hard to 
 \FPT-approximate within any constant factor when $r \geq 3$.
\end{theorem}

We use Theorem~\ref{thm:min3lin-hard} to demonstrate that there exist finite commutative rings
such that $\minlin{2}{R}$ is \W{1}-hard to \FPT-approximate within any constant factor.
Consider the polynomial ring $\ZZ_p[\rx_1,\dots,\rx_m]/(\rx_1^2,\dots,\rx_m^2)$, 
i.e.\ the ring with coefficients from the field $\ZZ_p$ for some prime $p$
and
indeterminates $\rx_1,\dots,\rx_m$ with $\rx_1^2,\dots,\rx_m^2$ factored out.
An element $r \in R$ is thus a sum 
$r_{\sf unit}+\sum_{i=1}^m r_{x_i} \rx_i+\sum_{1 \leq i < j \leq m} r_{x_{ij}}\rx_i\rx_j$, 
where all coefficients $r_{\sf unit}, r_{x_1}, \dots$ are from $\ZZ_p$.

\begin{theorem}
  \label{thm:min2lin-hard}
  $\minlin{2}{R}$
  is \W{1}-hard to \FPT-approximate within any constant factor
  when $R=\ZZ_p[\rx_1,\dots,\rx_m]/(\rx_1^2,\dots,\rx_m^2)$, $p$ is prime, and $m \geq 2$.
\end{theorem}
\begin{proof}[Proof for $p=2$, $m=2$.]
  We present a proof for the minimal case $p=2$ and $m=2$,
  and remark that it generalizes to an arbitrary prime $p$ and $m \geq 2$
  in a straightforward way.
  Consider an almost homogeneous instance $(S, k)$ of $\minlin{3}{\ZZ_2}$.
  We create an instance $(S_R, k)$ of $\minlin{2}{R}$
  with the same parameter value as follows.
  Start by adding variable $x_R$ to $V(S_R)$ for every $x \in V(S)$.
  For every (crisp) equation of the form $x = i$ in $S$, where $i \in \ZZ_2$,
  add a (crisp) equation 
  \begin{equation}
    \label{eq:const}
    x_R = i    
  \end{equation} 
  to $S_R$.
  For every equation of the form 
  $a + b + c = 0$ in $S$,
  create a new variable $v_R$ in $V(S_R)$ and 
  add the following crisp equations to $S_R$:
  \begin{equation}
    \label{eq:ternary}
    \rx    \cdot v_R =  \rx\ry \cdot b_R, \qquad
    \ry    \cdot v_R =  \rx\ry \cdot a_R, \qquad
    \rx\ry \cdot v_R = -\rx\ry \cdot c_R.
  \end{equation}
  This completes the construction.

  Towards correctness, first consider an assignment
  $\alpha : V(S) \to \ZZ_2$ and define
  $\alpha' : V(S_R) \to R$ as $\alpha'(x_R) = \alpha(x)$ 
  for all $x \in V(R)$;
  for every variable $v_R$ in $V(S_R)$,
  let $a + b + c = 0$ be the corresponding equations,
  and set $\alpha'(v_R) = \rx \cdot \alpha(a_R) + \ry \cdot \alpha(b_R)$.
  We claim that $\cost(\alpha', S_R) \leq \cost(\alpha, S)$.
  Clearly, if $\alpha$ satisfies an equation
  of the form $x = i$, then $\alpha'$ satisfies
  the corresponding constraint $x_R = i$ in $S'$.
  Moreover, if $\alpha$ satisfies an equation
  of the form $a + b + c = 0$ in $S$, 
  then we claim that $\alpha'$ satisfies all three equations 
  in~\eqref{eq:ternary}.
  Indeed,
  \begin{align*}
    \rx \cdot v_R &= \rx \cdot (\rx \cdot \alpha(a) + \ry \cdot \alpha(b)) = \rx\ry \cdot \alpha(b), \\
    \ry \cdot v_R &= \ry \cdot (\rx \cdot \alpha(a) + \ry \cdot \alpha(b)) = \rx\ry \cdot \alpha(a), \text{ and} \\
    (\rx + \ry) \cdot v_R &= \rx\ry \cdot (\alpha(a) + \alpha(b)) = \rx\ry \cdot (-\alpha(c)) = -\rx\ry \cdot \alpha(c).
  \end{align*}
  
  For the opposite direction, let 
  $\beta : V(S_R) \to R$ be an assignment to $S_R$.
  Recall that every element $r$ in $R$
  is of the form $r_{\sf unit} + r_x \rx + r_y \ry + r_{xy} \rx\ry$,
  and consider the assignment 
  $\beta' : V(S) \to \ZZ_2$ defined as
  $\beta'(x) = \beta(x_R)_{\sf unit}$.
  We claim that $\cost(\beta',S) \leq \cost(\beta,S_R)$.
  Indeed, if $\beta$ satisfies a unary equation $x_R = i$,
  then $\beta'$ satisfies the corresponding equation 
  $x = i$ because $i \in \ZZ_2$.
  Now assume $\beta$ satisfies equations~\eqref{eq:ternary}.
  We claim that $\beta'$ satisfies $a + b + c = 0$.
  Observe that for every $r \in R$,
  we have $\rx\ry \cdot r = \rx\ry \cdot r_{\sf unit}$.
  Hence, it suffices to show that $\beta$ satisfies
  $\rx\ry \cdot (a_R + b_R + c_R) = 0$.
  By summing up the first two equations in~\eqref{eq:ternary},
  we deduce that $\beta$ satisfies
  \[
    (\rx + \ry) \cdot v_R = \rx\ry \cdot (a_R + b_R).
  \]
  Combined with the third equation in~\eqref{eq:ternary},
  we get that $\beta$ satisfies
  \[
    \rx\ry \cdot (a_R + b_R) = -\rx\ry \cdot c_R \implies \rx\ry \cdot (a_R + b_R + c_R) = 0.
  \]
\end{proof}

\section{Concluding remarks}
\label{sec:conclusion}

Our main algorithmic result is that $\mintwolinz{p^n}$ 
for every prime $p$ and integer $n$
is \FPT-approximable within factor $2$ 
for every prime $p$ and $n \geq 2$.
The exact parameterized complexity of 
this problem was left open in~\cite{Dabrowski:etal:soda2023},
and remains an intriguing question.
In case the answer is positive,
this would immediately improve
the approximation factor 
for $\mintwolinz{m}$ to $\omega(m)$ by Proposition~\ref{prop:product-approx}.
Since $\ZZ_m$ is just a direct sum of $\omega(m)$ different rings $\ZZ_{p_i^{n_i}}$,
so that each equation in $\ZZ_m$ decomposes into $\omega(m)$ independent equations over these component rings,
such an approximation factor seems hard to beat. 
It is evident from the example illustrated in 
Figure~\ref{fig:classassignment}
that the greedy approach which dismantles
bundles of edges corresponding to an equation
inevitably leads to sub-optimality.


\bibliography{references}

\end{document}